\documentclass[12pt]{article}
\usepackage{amsthm,amsmath,amssymb,url,graphicx,color}
\usepackage{grffile}
\usepackage{tikz}
\usepackage{hyperref}

\hypersetup{
  colorlinks = true,
  urlcolor = blue,
  linkcolor = blue,
  citecolor = blue}

\theoremstyle{definition}
\newtheorem{thm}{Theorem}[section]
\newtheorem{prop}[thm]{Proposition}
\newtheorem{lem}[thm]{Lemma}

\newtheorem{ex}[thm]{Example}

\theoremstyle{remark}
\newtheorem{rem}[thm]{Remark}

\newcommand\cref[1]{Corollary~\ref{#1}}

\numberwithin{equation}{section}

\newcommand{\newoperator}[2]{\DeclareMathOperator{#1}{#2}}
\newoperator{\supp}{supp}
\newoperator{\spec}{Spec}
\newoperator{\tr}{Tr}
\newoperator{\re}{Re}
\newoperator{\im}{Im}
\newoperator{\sgn}{sgn}
\newoperator{\per}{per}
\newoperator{\var}{var}
\newoperator{\cov}{cov}
\newoperator{\rank}{rank}
\newoperator{\diag}{diag}
\def\:{\, : \, }
\def\;{\, ; \, }
\def\bra{\langle}
\def\ket{\rangle}

\def\la{{\mathbf{\lambda}}}
\def\La{\Lambda}
\def\L{\Lambda}
\def\B{\mathbf{B}}
\def\phi{\varphi}

\def\b{\beta}
\def\cI{\mathcal{I}}
\def\d{{\mathrm d}}
\def\e{{\mathcal E}}

\def\p{{\mathcal P}}
\def\s{{\sigma}}
\def\y{{\mathcal Y}}
\def\r{{\mathbb R}}

\def\F{{\mathbb F}}
\def\K{{\mathbb K}}
\def\Z{{\mathbb Z}}

\def\c{{\mathcal C}}
\def\n{{\mathbb N}}

\def\t{\theta}
\def\eps{\varepsilon}
\def\P{\mathbb{P}}
\def\E{\mathbb{E}}
\def\R{\mathbb{R}}
\def\g{\gamma}
\def\tr{\mathrm{Tr}}

\def\supp{\mathrm{supp}}

\def\th{\theta}
\def\SINR{\mathsf{SINR}_{o}}
\def\ze{\mathbf{0}}
\def\D{\mathbf{\Delta}}
\def\U{\text{U}}
\def\SNR{\text{SINR}}
\def\ol{\overline}
\def\xx{\mathbf{x}}

\makeatletter
\renewcommand*{\@cite@ofmt}{\hbox}
\makeatother

\newcommand{\numberthis}{\addtocounter{equation}{1}\tag{\theequation}}

\begin{document}
\title{{\Large Disordered complex networks : \\  energy optimal lattices and persistent homology}}
\author{
S. Ghosh \thanks{Dept. of Mathematics, National University of Singapore, \texttt{subhrowork{@}gmail.com}} 
\and
N. Miyoshi \thanks{Dept. of Math. \& Comp. Sc., Tokyo Instt. of Tech.,  \texttt{miyoshi{@}is.titech.ac.jp}}
\and
T. Shirai  \thanks{Institute of Math. for Industry, Kyushu University, \texttt{shirai{@}imi.kyushu-u.ac.jp}
} 
}
\date{}

\maketitle

\begin{abstract}
Disordered complex networks are of fundamental interest in statistical physics, and they have attracted recent interest as stochastic models for information transmission over wireless networks. 
While mathematically tractable, a network based on the regulation Poisson point process model offers  challenges vis-a-vis network efficiency.
Strongly correlated alternatives, such as networks based on
 random matrix spectra (the Ginibre network), 
on the other hand offer formidable challenges in terms of 
tractability and robustness issues. In this work, we demonstrate that network models based on random perturbations of Euclidean lattices \textit{interpolate} between Poisson and rigidly structured networks, and allow us to achieve the \textit{best of both worlds} :  significantly improve upon the Poisson model in terms of network efficacy measured by the \textit{Signal to Interference plus Noise Ratio} (abbrv. SINR) and the related concept of \textit{coverage probabilities}, at the same time retaining a considerable measure of mathematical and computational simplicity and robustness to erasure and noise. 

We investigate the optimal choice of the base lattice in
 this model, connecting it to the celebrated problem
 optimality of Euclidean lattices with respect to the
 Epstein Zeta function, which is in turn related to notions
 of lattice energy. This leads us to the choice of the
 triangular lattice in 2D and face centered cubic lattice in
 3D, whose Gaussian perturbations we consider. We provide theoretical analysis and empirical investigations to demonstrate that the coverage probability decreases with increasing strength of perturbation, eventually converging to that of the Poisson network. In the regime of low disorder, our studies suggest an approximate statistical behaviour of the coverage function near a base station as a log-normal distribution with parameters depending on the Epstein Zeta function of the lattice, and  related approximate dependencies for a power-law constant that governs the network coverage probability at large thresholds.

In 2D, we determine the disorder strength at which the
perturbed triangular lattice (abbrv. PTL) and the
 Ginibre networks are the \textit{closest}
 measured by comparing their network topologies via a
 comparison of their \textit{Persistence Diagrams} in the
 total variation as well as the symmetrized nearest
 neighbour distances. We demonstrate that, at this very same
 disorder, the PTL and the Ginibre networks exhibit very similar coverage probability
 distributions, with the PTL performing at least as well as
 the Ginibre. Thus, the PTL network at this
 disorder strength can be taken to be an effective
 substitute for the Ginibre network model, while at the same time offering the advantages of greater tractability both from theoretical and empirical perspectives.
\end{abstract}
\tableofcontents

\section{Introduction and main ideas}
\subsection{Stochastic spatial networks}
\subsubsection{Spatial networks and wireless communications}
The study of complex networks to understand and enhance wireless communication has attracted considerable interest in recent years. An important driving force behind this trend has been the exponentially growing volumes of data that are being generated and gathered, and the necessity of physical infrastructure to make the communication and collection of such data practicable.

An important feature of wireless communication networks is their spatial nature (\cite{dettmann2018spatial,barthelemy2011spatial,andrews2010primer}). Namely, the network consists of a large number of nodes that are distributed in space - typically 2 or 3 dimensional Euclidean space, although more exotic geometries have also been investigated (\cite{kleinberg2007geographic, yan2014accuracy}). This endows such networks with an inherent structure - e.g.,  in 2D they can be studied as (weighted) planar graphs, and therefore inherit all the characteristics possessed by such special classes of mathematical structures. Each of the nodes, or base stations, broadcast signals that interfere with each other; and an important objective is to understand how the field of signal strength at various locations, adjusted with the interferences, looks like across the ambient space. A particular goal would be to design network layouts that optimize such signal strength for most (or typical) locations.

\subsubsection{Disordered complex networks}
In the study of large  complex systems, a classical ansatz, particularly in the physical sciences, is to consider an analogous random system. The typical behaviour of the random system is believed to provide a good insight for how the large complex system will behave, particularly in terms of its relatively simple, readily measurable characteristics. A celebrated example of this approach is the famous random matrix model proposed by Wigner (\cite{wigner1955characteristic}) to understand the behaviour of large and complex nuclei. It turns out that the distribution of energy levels of the nuclei of large and complex atoms are captured extremely well by the spectral properties of random matrices, a phenomenon referred to as the \textit{Wigner surmise} (\cite{wigner1967random,mehta2004random}).   

 In the setting of large, complex networks, this approach can be executed via spatial random point process models. The use of random point processes to study the spatial distribution of wireless network models has been a popular topic in the recent years (\cite{baccelli2010stochastic, blaszczyszyn2018stochastic, ram2007path}). In these models, random points sampled from an appropriate underlying distribution are thought of as representing the locations of wireless nodes. Various objects of interest, like the signal to noise ratio, are studied as random variables, and their asymptotic behaviour in the limit of the system size tending to infinity analysed in order to understand its efficacy as a communication network.


\subsubsection{Poissonian networks and spatial independence}
The most popular and widely studied model in this respect is the Poisson point
process (\cite{baccelli2010stochastic,
blaszczyszyn2018stochastic,haenggi2012stochastic}). A key
feature of the Poisson point process is that the statistical
distributions of the points in disjoint domains are statistically
independent (\cite{kallenberg2006foundations}). This makes the Poisson point process the analogue of \textit{pure
noise} in the world of point processes. Another consequence of this
property is that it makes it rather convenient to do computations with
the Poisson point process - indeed, closed form expressions can be
derived for almost any statistic of interest, and very often asymptotics
can be studied in considerable detail.

\begin{figure}[htbp]
\begin{center}
\includegraphics[scale=0.65]{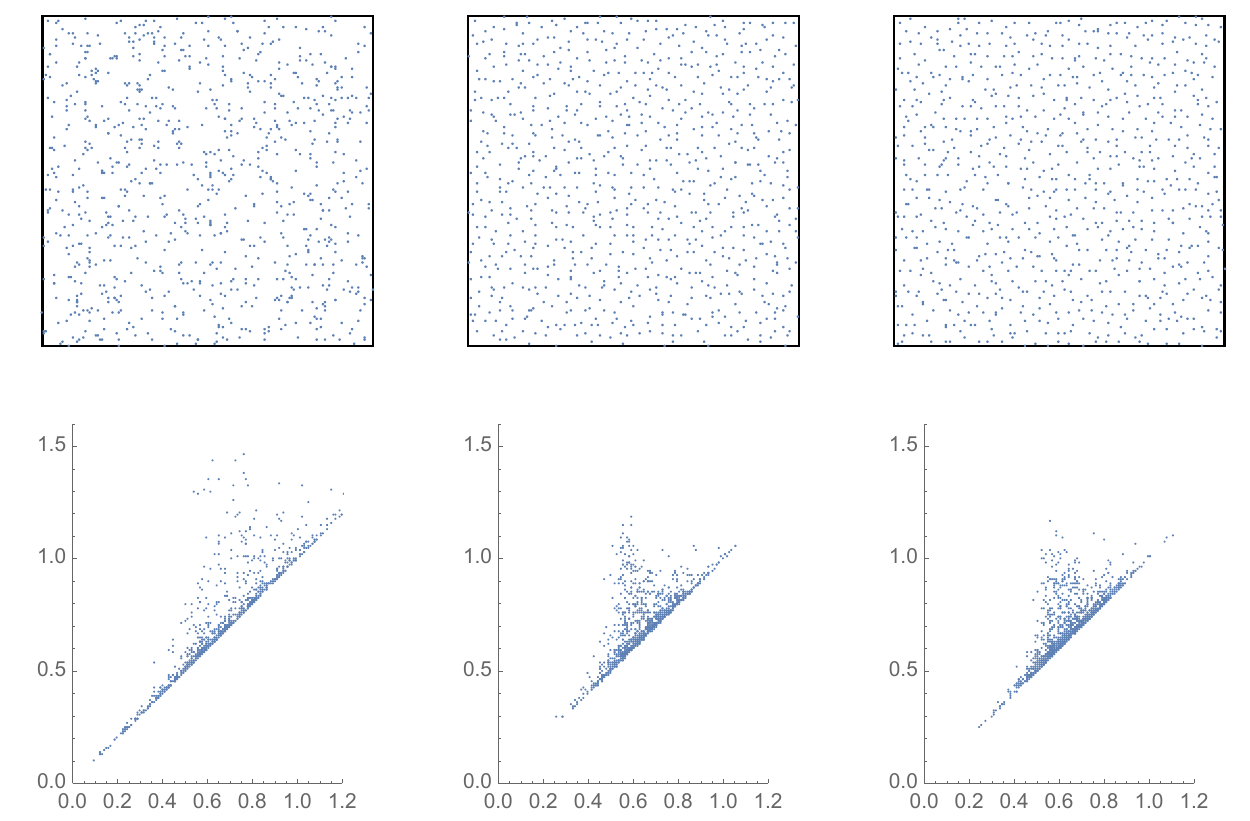}
\end{center}
 \caption{Upper panel: from left to right, simulations of Poisson, Ginibre and Perturbed Triangular Lattice (PTL) ($\sigma=0.4$) with unit intensity. 
Lower panel: the persistence diagrams of the above. }
 \label{fig:PP}
\end{figure}

However, the very same property of statistical independence becomes a
limitation if we think of the real-life modelling perspective. Indeed, the
spatial independence implies that, given that there is a node at a
location, the next node to be introduced is equally likely to occur
very close to or far away from it. This leads to the formation of
clusters of points in the Poisson point process, and to compensate for such
clusters since the average number of points is held fixed at some
constant, we have large vacant spaces devoid of any Poisson points
(see Fig. \ref{fig:PP} top left). Intuitively, this reduces the
efficacy of the Poisson point process as a model for wireless network, making
the coverage somewhat non-uniform and leading to an unnecessary excess
of coverage in certain patches and lack thereof in certain others. This
can, in fact, be established rigorously by comparing certain relevant statistics
of the Poisson point process with those of competing models (\cite{HKPV}).

\subsubsection{Spatial dependence and Ginibre networks}
This motivates the study of models
of random point sets which are bereft of these difficulties. The point process
to be used as a model must satisfy two basic criteria, which are often
somewhat contradictory in spirit. First of all, the point process must
allow sufficiently strong spatial correlation so as to induce mutually
repulsive behaviour at local scales and thereby preclude the clumping
behaviour as the Poisson point process which renders the latter rather
ineffectual as a model for wireless networks. On the other hand, it must
have sufficiently nice properties as a mathematical model so that key
statistical quantities can be estimated and theoretical or numerical analyses can be
carried out effectively. A related criterion, which is important from an
application-oriented perspective, is that the point process should be
easy to simulate, so that large scale statistical behaviour can be
gleaned from simulations, in case theoretical computations do not yield
sufficiently tractable expressions (which is often the case in
non-Poisson situations).

An important model of point processes which has been
studied in this respect is the so-called Ginibre ensemble
(\cite{MS,deng2014ginibre, LBDA14}), which can also be looked upon as the 2D Coulomb gas (or the 2D One Component Plasma) at the inverse temperature $\beta=2$ (\cite{Survey,jancovici1993large}).
Introduced by the physicist Ginibre as the non-Hermitian analogue of the famous Wigner Gaussian random matrix ensemble (\cite{ginibre1965statistical}), it
belongs to the special class of point processes known as
\emph{determinantal point processes} (or DPPs) (\cite{ST03,
HKPV}). In the rest of this article, we will occasionally
refer to it as the Ginibre ensemble
and the resulting
network as the \textit{Ginibre network}.  Via the connection to either random matrices or DPPs, it may be seen that mutual repulsion between the points is a built-in feature of the Ginibre network (\cite{HKPV}). Moreover, for quantities which involve only the absolute
values of the points like the coverage probability or link
success probability, we can exploit the fact
the process of absolute values has an equivalent description as
independent gamma random variables (\cite{HKPV}), which enables some facility with theoretical
analysis.


\subsubsection{Hyperuniformity and spatial networks}
Qualitatively, one can think of \textit{regularity} in a point process as relatively \textit{uniform} distribution of points over space as opposed to a uniform distribution in a probabilistic sense, arising out of the strength of spatial correlations - e.g., the tension between local repulsion pushing points away from each other, and a constraint on their average density being held fixed. In the recent literature, such behaviour have been studied in its own merit, under the broad umbrella of \textit{hyperuniformity}, and related \textit{rigidity phenomena}.  

\textit{Hyperuniformity}, also referred to as
\textit{superhomogeneity}, is the phenomenon of suppressed
fluctuations of particle numbers - in Poissonian systems,
the variance of the number of points in a large domain of
space grows like the volume (referred to in the physics
literature as \textit{extensive fluctuations}), whereas for
hyperuniform systems, the fluctuations are smaller order of
the volume (resp., \textit{sub-extensive}), often growing
only like the surface area of the domain
(\cite{torquato2003local, Survey}). This is exhibited by
many important systems in nature, a key example being that
of random matrix ensembles including the Ginibre ensemble discussed
above,  Coulomb systems in general, as well as more exotic
models like zeros of random functions (\cite{HKPV, Survey})
and \textit{stealthy hyperuniform systems}
(\cite{torquato2015ensemble,torquato2018hyperuniform,ghosh2018generalized}).
An alternate description of hyperuniformity can be obtained
via its ``structure function'' or ``power spectrum''
$S(k)$. Roughly speaking, it is the Fourier transform of the
(truncated) two point correlation function of the point
process. Hyperuniformity entails that $S(k) \to 0$ as $k \to
0$, and indeed an array of hyperuniform  behaviour can be obtained depending on the rate at which $S(k) \to 0$ as $k \to 0$ (\cite{torquato2003local,Survey,baake2015diffraction}).


\subsection{A three-fold investigative framework}
\label{sec:3-fold}

The above considerations open up the avenue for  investigation from three
perspectives. First of all, while the Ginibre network is more spatially regular
than the Poisson network, it is not the only such model, and
therefore it is a natural question to investigate whether
other models can provide additional value from the network
design perspective, while retaining the benefit of a
relatively regular point pattern. In particular, considering
hyperuniformity as a cornerstone of a spatially regular
point pattern, we can investigate  models of random point
sets which exhibit similar hyperuniform behaviour as the
Ginibre point process, but offer other advantages from the modelling and analysis points of view.

Secondly, in the Ginibre network,
the lack of
independence of any sort other than the distances of the points from the origin turns out to be a 
hindrance for theoretical analysis as well as computational investigations.

Finally, the Ginibre ensemble is more or less a stand-alone model, and it is difficult to introduce a
rich class of parameters so as to leave open the possibility of tuning
the model to data. From a statistical modelling
point of view, a parametric family would thus be of great interest. Another specificity of the Ginibre ensemble is its planarity, with no natural extensions to higher dimensions, wherein such network models would be significant particularly in 3D space.

In summary, it would be of great interest to investigate
point process models of wireless networks that embody the spatial regularity features of the
random matrix ensemble, while at the same time retain the
benefits of the Poissonian networks with regard to an independent latent structure for computation, simulation and analytical purposes, and to be able to do so in a parametric manner in general dimensions.

\subsection{Disordered lattices}
\label{sec:dislat}
In this article, we study a class of models which are promising
candidates for attaining many of objectives outlined
above. To this end, we consider a lattice $\L \subset \R^d$
for any arbitrary dimension $d$, and a mean-zero random
variable $\g$ with distribution $F$ on $\R^d$. We can then
define the point process
\[
 \p_{\L,\g}
= \{n + \gamma_n : n \in \L, \text{$\gamma_n$ i.i.d. copies of $\gamma
\sim F$}\}.
\]
Here a subset $\La \subset \R^d$ is said to be a lattice in $\R^d$ if $\La$
is an additive group of $\R^d$ which is expressed as 
\[
 \La = \{\sum_{i=1}^d a_i \mathbf{e}_i : a_i \in \mathbb{Z} \ (i=1,2,\dots,d)\}
\]
by using a basis $\mathbf{e}_1, \dots, \mathbf{e}_d$ of the vector
space $\R^d$. 
A particularly interesting case of this arises when $\g$ is Gaussian on $\R^d$ with mean 0 and variance $\s^2$, entailing that $\s$ automatically becomes a tuning parameter in the model. Another important consideration is the choice of the lattice $\L$, which, as we shall see in Section \ref{sec:energy-optimality}, will have a significant impact on the Signal to Interference plus Noise Ratio and will lead to interesting mathematical connections with energy optimality. Finally, as demonstrated in Section \ref{sec:conv-Poisson}, disordered Gaussian lattices interpolate continuously between the original lattice and the Poisson point process with the same density, as the noise parameter $\s$ varies. This provides us the opportunity to tune the parameter $\s$ so as to achieve a desired balance between structure and randomness in the network.

The lattice group $\La$ acts on $\R^d$ by translations
and its quotient $T_{\La} :=\R^d / \La$ turns out to be
a torus, which is represented by its fundamental parallelotope of $\La$ 
\[
 \mathcal{D}_{\La} = \{\sum_{i=1}^d a_i \mathbf{e}_i : a_i \in [0,1) \ (\forall i
 =1,2,\dots,d)\}. 
\]
If one desires to have translation invariance of the perturbed
lattice, one can add a uniform shift by adding a uniform
random vector from $\mathcal{D}_{\La}$. 
 It turns out that such $ \p_{\L,\g}$ is hyperuniform as soon as the ``tail'' of $\g$ (i.e., $\P[|\g|>t]$) decays faster than $t^{-d}${, where $|\gamma|$ is the Euclidean norm in $\R^d$}. For a theoretical discussion of this property, we refer the reader to \cite{Survey, ghosh2017number}. 
Since the effect of the uniform shift on SINR appears just as the
integration over $\mathcal{D}_{\La}$, for simplicity of
our discussions, we will omit it in the present paper. 
%

The above models, in spite of exhibiting \textit{rigid structure} or
\textit{regularity}, as manifested in their hyperuniform behaviour,
 have an independence structure clearly built into
them. The fact that the particles tracing their origin to different
lattice points in $\L$ are independent enables us to write down closed
form expressions for various statistics of interest, and also
facilitates an ease of simulation that is not available with the Ginibre
random network.

Although we have introduced the models
with i.i.d. perturbations, an interesting variant is one
where the perturbations are independent but not identically distributed,
by varying the scale parameter (i.e. the standard deviation of $\g$), or for that matter,
considering different random variables for perturbing different lattice
locations. This can be utilized to factor in spatial inhomogeneity of
a wireless network  node distribution.

One can, for example, consider a statistical problem of fitting
a lattice perturbation model as above to a given wireless network, by
estimating parameters like the standard
deviations of Gaussian fluctuations in different lattice sites, or if one is more ambitious, estimating the statistical properties of an unknown distribution of perturbation that might belong to a wider family like the exponential family. It
is very difficult to formulate a reasonable analogue of such questions
in the context of other strongly correlated models. 

Finally, an important advantage of perturbed lattice processes in modelling wireless base stations is that these processes are robust to missing data and erasure. To be more precise, if a collection of perturbed lattice 
points are missing or corrupted, they can be re-generated by simply re-sampling the perturbations corresponding to those particular lattice points, with the rest of the configuration being left unchanged, and this procedure fully preserves the statistical properties of the random network. This is an important practical advantage that sets disordered lattice processes apart from other  correlated point processes like random matrix processes, where erasure or corruption of a subset of points do not admit any simple correction.


\subsection{The network observables}
\label{sec:netobs}

\subsubsection{The basic set-up}
We are now ready to define the key observables of the
induced network that we are going to study, which we will do
in the completely general set up of an arbitrary point
process. 

Our setting is the following.
A configuration $\Phi = \sum_{i=1}^{\infty} \delta_{X_i}$
is a simple (stationary) point process on $\r^d$
and it provides a realization of spatial configuration of
base stations of a cellular network.
A decreasing function
$\ell : (0,\infty) \to [0,\infty)$ is a \textit{path-loss function},
which represents the attenuation of signals at distance $r$.
A random variable $F_i$, independent of the point process $\Phi$,
represents a random effect of fading/shadowing
from the base station $X_i$ to the typical user.
Here we assume the so-called \textit{Rayleigh fading}, i.e.,
$\{F_i\}_{i=1}^{\infty}$ are i.i.d. exponential random variables with
mean $1$. Let $W$ be a random variable representing \textit{thermal noise} (modelling general random disturbances from the \textit{environment}),
independent of $\{F_i\}_{i=1}^{\infty}$ and $\Phi$.

\subsubsection{The SINR and its distribution}
Suppose that a typical user is located at the origin (since
our point processes are translation invariant, statistically there is
no loss of generality in reducing to the origin as the point of reception), and
 is receiving the signal associated with the nearest base station $X_B$ from the
origin, where $B$ is the lattice site corresponding to the nearest base
station. This signal is being retarded by interference from other base stations, and by the pure thermal noise $W$.
 The Signal-to-Interference-plus-Noise-Ratio (henceforth abbreviated as SINR) at the origin is defined by
\begin{equation} \label{eq:SINR_0}
 \SINR = \frac{F_B \ \ell(|X_B|)}{W + I(B)}
\left(= \frac{\text{signal}}{\text{noise}}\right),
\end{equation}
where $I(B) = \sum_{i \not= B} F_i \ \ell(|X_i|)$ is the cumulative
interference signal from all the base stations other than $B$.  $\SINR$ is the observable by which we are going to adjudicate
the efficiency of the network, and hence this quantity will of paramount interest in our considerations. For detailed considerations on this model, including its motivational origins and effectiveness, we refer the interested reader to  (\cite{baccelli2010stochastic,MS}).

Notice that, because of the randomness in the locations of the base stations, $\SINR$ is a random variable. In order to compare $\SINR$ for two random networks, we compare their tails, that is, the probability that the $\SINR$ exceeds a certain level $\theta$. This probability is known as the \textit{coverage probability} (corresponding to the level $\theta$), and greater the coverage probability for a given $\theta$, better is the network.

In the set-up of signal, interference and noise discussed above, the coverage probability is given by

\begin{prop}[Proposition~2.2 in \cite{MS2}]
\label{prop:SINR}
Suppose that base stations are distributed according to
a simple point process $\Phi = \sum_{i=1}^{\infty} \delta_{X_i}$.
Then, the coverage probability is given by the formula
\begin{equation} \label{eq:prop_SINR}
P(\SINR > \theta)
= E\left[\prod_{j \not= B} \left(1 + \theta
 \frac{\ell(|X_j|)}{\ell(|X_B|)} \right)^{-1}
\right],
\end{equation}
where $X_B$ is of the least modulus among the base station locations.
\end{prop}

\subsubsection{Rayleigh fading and related effects}
In particular, in the important case when $\ell(r) = a r^{-2\beta} \ (\beta > 1)$ and the spatial dimension $2$, the coverage probability $p_c(\theta, \beta)$ is given by
\begin{equation} 
p_c(\theta, \beta)
:= P(\SINR > \theta)
= E\left[
\prod_{j \not= B}
\left(1 + \theta
 \left|\frac{X_B}{X_j}\right|^{2\beta} \right)^{-1}
\right].
\label{eq:infiniteprod}
\end{equation}
In general dimension $d \ge 2$, a natural choice for the path-loss is $\ell(r) = a r^{-d\beta} \ (\beta > 1)$, and it turns out that 
\begin{equation} \label{eq:infiniteprod_d}
p_c(\theta, \beta)
=E\left[\prod_{j \not= B} \left(1 +\theta
 \frac{|X_B|^{d\beta}}{|X_j|^{d\beta}}  \right)^{-1}
\right].
\end{equation} 
In this work, as also with natural applications, we mostly concern ourselves with dimensions 2 and 3.

\begin{rem} \label{rem:SINR_a}
More generally, we can consider the SINR at a general point $a$, possibly different from the origin. This will be given by 
\begin{equation} \label{eq:SINR_a}
 \text{SINR}_a = \frac{F_B \ \ell(|X_B-a|)}{W + I(B;a)}
\left(= \frac{\text{signal at $a$}}{\text{noise at $a$}}\right),
\end{equation}
where $X_B$ is the nearest base station to $a$, and $I(B;a) = \sum_{i \not= B} F_i \ \ell(|X_i-a|)$ is the cumulative
interference signal from all the base stations other than $B$, observed at $a$. Expressions analogous to \eqref{eq:prop_SINR} and \eqref{eq:infiniteprod} can also be obtained. In particular, we have the following expression for the coverage probability at the location $a$:
\begin{equation} \label{eq:infiniteprod_d_a}
p_c^{[a]}(\theta, \beta)
= E\left[\prod_{j \not= B} \left(1 +\theta
 \frac{|X_B-a|^{d\beta}}{|X_j-a|^{d\beta}}  \right)^{-1}
\right],
\end{equation} 
where $B$ is the nearest base station to $a$.
In particular, this focusses attention on the \textit{coverage function} at the location $a$. For a point configuration $\Phi$ obtained by a random perturbation of the lattice $\La$, the coverage function at $a$ is the random variable
\begin{equation} \label{eq:cov_func_a}
\mathcal{C}(\Phi;a)=\prod_{B \not= j \in \La} \left(1 +\theta
 \frac{|X_B-a|^{d\beta}}{|X_j-a|^{d\beta}}  \right)^{-1}.
 \end{equation}
\end{rem}

\subsubsection{SINR for disordered lattices}
It can already be understood from Eq.~\eqref{eq:infiniteprod} why negatively dependent (i.e., mutually repelling) point configurations would be effective in improving the coverage probability. To this end, we focus on the terms
$\left(1 + \theta  \left|\frac{X_B}{X_j}\right|^{2\beta} \right)^{-1}$. In order that the value of such a term be high, two things would be conducive. First, $|X_B|$ should be preferably low. Secondly, the most important terms that can damp the coverage probability are those for which $|X_j|$ is small subject to the constraint that $|X_j|>|X_B|$, and the contribution of these terms should be not too small, which essentially requires that there are not too many $X_j$ that are  farther than $X_B$ from the origin but not too far. This necessitates that points do not cluster close to the sphere in $\R^d$ on which the nearest base station is located.

For repulsive point processes, typically  the nearest point
to the origin is closer to the origin with a higher
probability. This is corroborated by the fact that the
\textit{hole probability} for radius $R$ (i.e., the
probability of having no points inside a ball of radius $R$)
typically decays faster than the Poisson point process (\cite{HKPV}). For the nearest base station to be far away from the origin, there has to be a big \textit{hole} centred at the origin, which is statistically unfavourable  in repulsive point processes. Furthermore, negatively dependent processes also statistically discourage clustering of points in a region  of space. These two properties of repulsive point processes help in improving the coverage probability, and make them ideal candidates for base station distribution in wireless networks.

As discussed in Section \ref{sec:dislat}, our main focus in this work is on random networks where the base stations are distributed as a disordered lattice. Let $\Lambda$ be a lattice in $\R^d$. 
We consider the following probability density function on $t \ge 0$, indexed by $n \in \Lambda$:
\begin{equation} \label{eq:densitydef}
f(t,n,\sigma)=e^{-\frac{|n|^2}{2\sigma^2}}  \cdot t^{\frac{d}{2}-1}e^{-\frac{1}{2}t} \cI_d(\sigma^{-1}|n|\sqrt{t}),\quad (t\ge 0),
\end{equation}
where $|n|$ denotes the Euclidean norm of $n$ in $\R^d$ and 
\[
\cI_d(u)=\frac{1}{2\cdot (2\pi)^{d/2}} \int_{\mathbb{S}^{d-1}}
e^{u \langle \omega , e_1 \rangle} \d \omega,
\] 
with $e_1$ being the first standard co-ordinate vector in $\R^d$ and $\d \omega$ being the standard spherical measure on $\mathbb{S}^{d-1}$.

\begin{rem}
For $d=2$, the function $\cI_d$ turns out to be closely related to the modified Bessel function of the first kind, given by (cf. \cite{Lebedev})
\[
 I_0(z) = \frac{1}{2\pi} \int_0^{2\pi} e^{z \cos \phi} \d\phi, \quad (\re z>0).
\]
For general $d$, it is easy to see that  
\begin{equation}
\cI_d(u) = \frac{1}{2^{d/2} \pi^{1/2} \Gamma(\frac{d-1}{2})} 
\int_0^{\pi} e^{u \cos \phi} (\sin \phi)^{d-2} \d\phi. 
\label{eq:generalI} 
\end{equation}
\end{rem}

Then we can state
\begin{thm}
 \label{thm:SINR}
\begin{itemize}
\item[(i)] The coverage probability is given by
\begin{align}
 p_c(\theta, \beta, \sigma)
&= \sum_{n \in \Lambda}
\int_0^{\infty}
\bigg\{\prod_{j \not= n \in \Lambda}
\int_{t }^{\infty}
\Big(1 + \th \frac{t^{d\beta/2}}{u^{d\beta/2}} \Big)^{-1}
 f(u, j, \sigma) \d u \bigg\} f(t, n, \sigma) \d t.
\end{align}
\item[\rm (ii)] The limit $C_1(\beta, \sigma)
=  \lim_{\th \to \infty} \th^{\frac{1}{\beta}} p_c(\th, \beta, \sigma)$
exists and \[
	C_1(\beta, \sigma) =\sum_{n \in \Lambda} e^{-\frac{|n|^2}{\sigma^2}}\int_0^\infty \bigg\{ \prod_{j \ne n \in \Lambda} \int_0^\infty \Big(1 + \frac{s^{d\beta/2}}{u^{d\beta/2}} \Big)^{-1} f(u,j,\sigma) \d u \bigg\} \frac{s^{\frac{d}{2}-1}}{2^{d/2}\Gamma(d/2)} \d s.
\]
\end{itemize}
\end{thm}
A key implication of Theorem \ref{thm:SINR} is that, for fixed $\beta$ and $\sigma$,  the curve $p_c(\theta,\beta,\sigma)$ v.s. $\theta$ is asymptotically a power law, and therefore, for large values of $\theta$, improving the coverage probability amounts to designing networks that provide a bigger value for $C_1(\beta,\sigma)$, which is purely a lattice-dependent quantity for a given level of disorder $\sigma$.

We defer the proof of Theorem \ref{thm:SINR}, along with the statement and proof of an auxiliary lemma, to Section \ref{sec:thmproof}.

\subsection{Main results and contributions}
Herein, we discuss the main results and contributions obtained in this paper. 

\textbf{Disordered lattices as spatial random network
models.} A central theme in this work is to demonstrate that
disordered lattice models (i.e. network models based on
random perturbations of Euclidean lattices) as highly
effective models for random spatial networks, covering in
particular applications to wireless network models. We
demonstrate that disordered lattices interpolate between
Poisson and Ginibre networks, and allow us
to achieve the best of both worlds : significantly improve
upon the Poisson model in terms of network efficacy measured
by the  SINR, at the same time retaining a considerable
measure of mathematical and computational simplicity and
robustness to erasure and noise. Our approach is
substantiated via theoretical analysis as well as empirical
investigations. Detailed comparisons of performance to
Poissonian and Ginibre networks are carried out in Section \ref{sec:DisLat_comp}. Furthermore, we prescribe the optimal lattice and the optimal level of disorder for perturbed lattice models in the context of wireless network applications. These are explored in detail in Sections \ref{sec:energy-optimality} and \ref{sec:DisLat_comp} respectively. The  closest approximation to a Ginibre point process by a disordered lattice model
 is obtained at the disorder level $\s \sim 0.4$, which interestingly  is also the disorder level at which the coverage probability distributions of the two network models roughly match. In a nutshell, we put forward disordered lattice models as a viable paradigm to  answer the questions alluded to in the three-fold investigative framework of Section \ref{sec:3-fold}. 

\textbf{Coverage probability for disordered lattice and  power-law asymptotics.} We obtain an explicit expression for the coverage probability for a general disordered lattice which, albeit in the form of an infinite sum of infinite products, provides a non-random object that can be investigated numerically. In fact, it is explicit enough to demonstrate that, in the regime of large threshold $\t$, the coverage probability is asymptotically a power law. This is considered in Theorem \ref{thm:SINR}.

\textbf{Optimal lattices  and connections to the Epstein
Zeta function.} We connect our search for optimal lattices
in the context of perturbed lattice models to the celebrated
Epstein Zeta function of a lattice, a kind of lattice energy
that is of intrinsic   interest in number theory and other
braches of pure mathematics. This is taken up for detailed
consideration in Section \ref{sec:energy-optimality}. In
summary, our theoretical investigations  suggest that the
behaviour of SINR (and coverage probabilities) is optimised
by considering perturbations of lattices that minimize the
Epstein Zeta function. In 2D, this leads to the choice of
the triangular lattice, whereas in 3D this suggests
considering the face centred cubic (abbrev. FCC) lattice. 

\textbf{Theoretical analysis of SINR in extremal regimes.} In wireless networks, the SINR is an object of central importance; in the setting of disordered networks it is a random variable. The SINR is notoriously  difficult to handle theoretically, evading a neat mathematical description, which makes its theoretical analysis complicated. However, in  this work, we obtain an approximate theoretical understanding of the SINR in  certain settings, focussing on regimes of small (or large) parameters, which has important implications for wireless networks (Section \ref{sec:small-parameter}). In particular, such understanding motivates our considerations for optimal lattice. 

In the regime of small $\s$,  we are able to obtain an approximation of the coverage function near a base station by a log-normal random variable with explicitly specified parameters (Section \ref{sec:small-sigma}). This leads to an approximation of the coverage probability by an explicit closed form integral that is at the same time simple enough to be numerically tractable.

 In the regime of small $\s$ and large $\t$, we carry out further theoretical analysis that suggests  explicit parametric dependencies of the coverage probability (c.f. Theorem \ref{thm:SINR}) in this regime (Section \ref{sec:large-theta}). 
Notably, our studies indicate an inverse dependence on $\s$ (in the form of an explicit power law) - an effect that is corroborated by our empirical investigations in Section \ref{sec:DisLat_comp}. It also suggests an inverse (power-law) dependence on the Epstein Zeta function of the lattice,  lending further credence to our choice of optimal lattice via a comparison of lattice energies.
 
 In the regime of small $\s$ and small $\t$, we unveil an approximation that entails linear dependence of the coverage probability near a base station on $\t$, and demonstrates worsening behaviour with increase in $\s$. Finally, our theoretical studies of the SINR indicate its monotonicity in $\s$ in the small $\s$ regime.

\textbf{A paradigm for measuring proximity between point sets.} In this work, we unveil a paradigm for measuring the proximity between point sets, which we believe would be of independent interest in a wide range of applications. Our approach considers the so-called \textit{Persistence Diagrams} of the point sets, and computes the \textit{Total Variation} distance between these (also exploring alternative possibilities such as the \textit{symmetrised nearest neighbour distance}). As elaborated in Section \ref{sec:PD}, the persistence diagram effectively captures the higher order geometry of a point set, making it a comprehensive and robust observable for this purpose.

It turns out (reference) that the closest approximation to a Ginibre random point set by a Perturbed Triangular Lattice in 2D is obtained near the disorder value $\s=0.4$, which is interestingly also the disorder level at which the SINR for this lattice model closely approximates (and slightly outperforms) the SINR for the Ginibre model. This enables us to suggest that a perturbed triangular lattice with disorder level $\s=0.4$ is an appropriate substitute for the Ginibre random network, while having the additional advantages of simplicity and robustness that are accorded by a disordered lattice model. 

We complement our investigations on measuring proximity between point sets by a much simpler observable to understand the geometry of a point set, which is  via its \textit{Nearest Neighbour Distribution}. Measuring distances between the nearest neighbour distance curves, while much cruder and less comprehensive that via their persistence diagrams, can nevertheless be considered in situations where computation simplicity is a greater consideration than  accuracy.

\textbf{Disordered lattice models interpolating with Poissonian networks.} We complete our investigations by rigorously demonstrating that disordered lattice networks interpolate lattices with Poissonian networks. This is entailed by a convergence of disordered lattices to the Poisson point process as random point configuration. In fact, we are able to demonstrate a result that holds in much greater generality; this is captured by Theorem \ref{result:conv} in Section \ref{sec:conv-Poisson}. The proof of Theorem \ref{result:conv} invokes classical theory of diffusions on infinite particle systems. This convergence is further  borne out  empirically by the convergence of the nearest neighbour distributions of the disordered lattice models to that of the Poissonian network (c.f. Fig.~\ref{fig:2dim-nearest-neighbour} and Fig.~\ref{fig:3dim-nearest-neighbour}).

\textbf{Comparison with Poissonian network.}  The crucial comparison between disordered lattice networks and the popularly used Poissonian network is that of the coverage probability distributions. It is demonstrated via our empirical investigations in Section \ref{sec:DisLat_comp} that for all values of disorder considered, the coverage probability for the disordered lattice has a better behaviour than Poisson at the same threshold $\theta$. This is indicated by the a higher value of the coverage function for the disordered lattices, which entails that the coverage function curve for those lie above that of the Poisson network (c.f. Fig. \ref{fig:2dim-SINR-2} and Fig. \ref{fig:3dim-SINR-2}).



\section{SINR in extremal parameter regimes} 
\label{sec:small-parameter}

\subsection{The regime of small $\s$: generalities} \label{sec:small-sigma-gen}

In this section, we will examine the approximate behaviour of the coverage probability in the small $\sigma$ regime, and explore the various consequences thereof in subsequent sections. We focus on the setting where the configuration $\Phi$ of base stations is a perturbation of the lattice $\La$ by the i.i.d. random variables $\{\s \xi_\la\}_{\la \in \La}$, with $\xi_\la$ being i.i.d. on $\R^d$ with unit standard deviation and $\s>0$ being the common standard error of the perturbations. For $\la \in \La$, set $X_\la = \la + \s \xi_\la$. 

For a point configuration $\Phi = \sum_{i=1}^{\infty} \delta_{X_i}$ and a given threshold $\t>0$, we will consider the coverage function at the location $a$, given by 
\begin{equation} \label{eq:covfunc}
\c(\Phi;a) := \prod_{\la \ne B}
\left(1 + \theta \frac{|X_B-a|^{d\beta}}{|X_\la-a|^{d\beta}} \right)^{-1},
\end{equation} 
where $\la$ ranges over $\La$ except $B$. This is in fact the coverage probability for given locations of the base stations (the randomness being in the fading), which is called the meta-distribution of SINR (\cite{HAE16}). 
In subsequent discussions, for a point configuration $\La$ and a point $\mathbf{a} \in \La$, we will denote by $\La_{\mathbf{b}}$ the point configuration consisting of all points of $\La$ except $\mathbf{b}$.

In Section \ref{sec:netobs} we considered translation invariant point processes as base station configurations. This would imply that the statistical  distribution of the SINR is the same at all points of space, and therefore it suffices to consider the SINR at the origin. Here, we consider an equivalent way of describing the same random variable. To this end, we let $\D$ to a  \textit{primitive unit cell} of the lattice containing the origin; e.g. for a triangular lattice in $\R^2$, it is an equilateral triangle with side length equalling the lattice spacing having the origin as a vertex. Let $\U(\D)$ denote the uniform distribution on $\D$. Let $\Phi$ be a perturbation of the lattice $\La$ as above, and let $\ol{\Phi}$ be the translation invariant version of $\Phi$. For a point process $\Phi$ on $\R^d$ and a location $p \in \R^d$, we use $\SNR_p(\Phi)$ to denote the random variable that is the SINR at $p$ for a base station configuration sampled from the process $\Phi$. 

With these notations, if $\xx \sim \U(\D)$ and statistically independent of $\Phi$, then the random variables $\SINR(\ol{\Phi})$ and $\SNR_\xx(\Phi)$ have the same statistical distribution. Thus, in order to understand $\SINR(\ol{\Phi})$, it is of interest to study the the random coverage function $\c(\Phi;\xx)$, in view of Proposition \ref{prop:SINR}.


In the regime of small $\s$, with high probability the nearest base station to $\xx$ will be the lattice perturbation of the one of the vertices of $\D$; we denote the latter by $B$. We recall from \eqref{eq:covfunc} that
\begin{align*}
\c(\Phi;\xx) 
:= & \prod_{\la \ne B} \left(1 + \theta \frac{|X_B-\xx|^{d\beta}}{|X_\la-\xx|^{d\beta}}\right)^{-1} \\
= & \prod_{\la \ne B} \left(1 + \theta \frac{|B+\xi_B-\xx|^{d\beta}}{|\la+\xi_\la-\xx|^{d\beta}} \right)^{-1}.
\end{align*}
Notice that $\xx-B$ is also uniformly distributed on a primitive unit cell of the lattice, and has the origin as the nearest lattice point. In view of this, we may focus attention to the case $B=\mathbf{0}$, the origin in $\R^d$. Thus, we are interested in the random variable 
 \begin{equation} \label{eq:cov_modified}
 \prod_{\la \ne \ze} \left(1 + \theta \frac{|X_{\ze}-\xx|^{d\beta}}{|X_\la-\xx|^{d\beta}} \right)^{-1}
= 
\prod_{\la \ne \ze} \left(1 + \theta \frac{|\s \xi_{\ze}-\xx|^{d\beta}}{|(\la-\xx) + \s \xi_\la|^{d\beta}} \right)^{-1} 
\end{equation} 
on the event that the closest lattice point to $\xx$ is the origin $\ze$. 

Observe that since $\xx$ is constrained to be closer to $\ze$ than other vertices of $\D$ in \eqref{eq:cov_modified}, there is an automatic bound on $|\xx|$ that places it close to the origin. A rich statistical behaviour arises, however, when $\xx$ is further constrained to be at the origin. In the interest of brevity, we focus on this detailed statistical structure in the present paper, postponing a more comprehensive analysis of the situation with general $\xx$ for future work. In other words, in the rest of this section, we will focus on 
\begin{equation}
\c_0(\Phi)
:= 
\prod_{\la \ne \ze} \left(1 + \theta \frac{|\s \xi_{\ze}|^{d\beta}}{|\la + \s \xi_\la|^{d\beta}} \right)^{-1} 
\end{equation} 

As we shall see, the statistical structure of $\c_0(\Phi)$ already suggests natural choices for the lattice $\La$ in 2 and 3 dimensions, which are the settings of greatest practical significance. From a modelling perspective, the focus on $\xx=0$ can be envisaged as a situation where we have a specially important point of interest (viz., the origin) and we intend to put a base station near that location, and want to focus attention on the SINR at that special point in the presence of the confounding effect of interference from farther base stations and ambient noise. 

\begin{rem} \label{rem:theory-vs-practical}
It may be noted that, although we are focussing at the origin in this section for  a detailed statistical examination of the SINR, the comparison of different point fields in this paper, such as in Section \ref{sec:DisLat_comp} are all with regard to the standard definition of $\SINR$ where the translation invariance of the point process has been taken into account. The discussions in this section are envisaged as an exploration of  some broad structural features of the perturbed lattice models as network distributions, rather than as rigorous theorems which would establish certain expected phenomena. 
\end{rem}

\subsection{The regime of small $\s$: a log-normal approximation} \label{sec:small-sigma}

We examine the logarithm of the coverage function near a base station
\[
\log \c_0(\Phi) 
= - \sum_{\la \in \Lambda_{\ze}} \log \left(1
+ \theta \frac{|\s \xi_\ze|^{d\beta}}{|\la + \s
\xi_\la|^{d\beta}} \right) 
= - \sum_{\la \in \Lambda_{\ze}}
\log \left(1 + \theta \s^{d\b} \frac{ |\xi_\ze|^{d\beta}}{|\la +
\s \xi_\la|^{d\beta}} \right) .  
\]
Since $\la \ne \ze$ and we are in the small $\s$ regime, we may expand the logarithm in a series as
\[\log \left(1 + \t \s^{d\b} \frac{ |\xi_\ze|^{d\beta}}{|\la + \s \xi_\la|^{d\beta} }\right) = \sum_{k=1}^\infty \frac{(-1)^{k+1}}{k} \t^k \s^{d\b k} \frac{ |\xi_\ze|^{d\b k}}{|\la + \s \xi_\la|^{d\b k}}.  \]
Therefore, we have for $\log \c_0(\Phi)$ the expansion
\begin{equation} \label{eq:expansion}
\log \c_0(\Phi) = \sum_{k=1}^\infty \frac{(-1)^k}{k} \t^k \s^{d\b k} \left( \sum_{\la \in \Lambda_{\ze}}\frac{ |\xi_\ze|^{d\b k}}{|\la + \s \xi_\la|^{d\b k}} \right). 
\end{equation}
As $\s \to 0$, the terms in the above expansion decay exponentially fast in $k$, so to the leading order in $\s$ we get
\begin{equation} \label{eq:leadorder}
\log \c_0(\Phi) = - \t \s^{d\b} \left( \sum_{\la \in
			       \Lambda_{\ze}} 
\frac{|\xi_\ze|^{d\b}}{|\la + \s \xi_\la|^{d\b}} \right) + O(\s^{2d\b}).  
\end{equation}
Observe that 
\begin{equation} \label{eq:quadexp}
|\la + \s \xi_\la|^{-d\b}=|\la|^{-d\b} \left| \omega_\la+\s \frac{\xi_\la}{|\la|} \right|^{-d\b}, 
\end{equation}
where $\omega_\la$ is the direction of the vector $\la \in \La$ (so that $\omega_\la$ is an element of $\mathbb{S}^{d-1}$). 

Then, in the small $\s$ regime, we can expand 
\begin{align*}
\left|\omega_\la+\s \frac{\xi_\la}{|\la|} \right|^{-d\b} 
&= \left( 1 + 2\s \frac{\langle \omega_\la, \xi_\la \rangle}{|\la|} + \s^2 \frac{|\xi_\la|^2}{|\la|^2} \right)^{- \frac{1}{2}d\b} \\
&=1 - d\s \b\frac{\langle \omega_\la, \xi_\la \rangle}{|\la|} + O(\s^2).  
\end{align*}
Combined with \eqref{eq:leadorder} and \eqref{eq:quadexp}, this implies that, to the leading order in $\s$ we have
\[ \log \c_0(\Phi) = - \t \s^{d\b} |\xi_\ze|^{d\b} \left[ \sum_{\la \in \Lambda_{\ze}} \left(\frac{1}{|\la|^{d\b}} - d \s \b \frac{\langle \omega_\la, \xi_\la \rangle}{|\la|^{d\b +1}}  \right) \right] + O(\s^{d\b+2}). 
\]
At this point, we recall the Epstein Zeta function of the lattice $\La$ (at the parameter $s$) as 
\begin{equation} \label{eq:epzeta}
\e_\La(s)=\sum_{\la \in \Lambda_{\ze}}\frac{1}{|\la|^s},
\end{equation} 
see, e.g., (\cite{titchmarsh1986theory,terras2012harmonic}). Using the Epstein Zeta function, we can express the leading order behaviour of the log coverage function near a base station as 
\begin{equation} \label{eq:logcov-Epzeta}
 \log \c_0(\Phi) =  - \t \s^{d\b} |\xi_\ze|^{d\b} \e_\La(d\b) +  d \t \b \s^{d\b + 1} |\xi_\ze|^{d\b} \left( \sum_{\la \in \La_{\ze}} \frac{\langle \omega_\la, \xi_\la \rangle}{|\la|^{d\b + 1}} \right) + O(\s^{d\b+2}).
\end{equation}

We now focus on the situation where, for a given location of the nearest base station, we are interested in the behaviour of the coverage function near a base station as the locations of the other base stations and the fading fluctuates randomly. In other words, we fix the \textit{signal}, and investigate the statistical effects of the \textit{interference} on the coverage probability. From the analysis presented above, it is clear that for a given location of the nearest base station $\xi_\ze$ and small $\s$, the coverage function near a base station $\c_0(\Phi)$ (equivalently, its logarithm) is maximised  when $\e_\La(d\b)$ is minimised.
This is the famous problem of finding the minimizing lattice for the Epstein Zeta function (\cite{SS}).

To make further analysis, we focus on the natural setting of the perturbations $\{\xi_\la\}_{\la \in \La}$ being $d$-dimensional standard Gaussians. It may be noted that, if $\xi_\la$ is a $d$-dimensional standard Gaussian, then $\eta_\la:=\langle \omega_\la, \xi_\la \rangle$ is a 1-dimensional Gaussian with mean 0 and variance $1$. As a result, it may be deduced that the random sum   \[\left(\sum_{\la \in \La_{\ze}} \frac{\langle \omega_\la, \xi_\la \rangle}{|\la|^{d\b + 1}} \right) \] is in fact a 1-dimensional Gaussian with mean 0 and variance $\e_\La(2d\b + 2)$. 

This implies that in the regime of small $\s$, the coverage function near a base station $\c_0(\Phi)$, for any given location $\xi_\ze$ of the nearest base station, is approximately a log-normal random variable (cf. \cite{johnson1995continuous}) with parameters
\begin{equation} \label{eq:log-normal-parameters}
- \t \s^{d\b} |\xi_\ze|^{d\b} \e_\La(d\b) \quad \text{ and  } \quad  \t \b \s^{d\b + 1} |\xi_\ze|^{d\b} \sqrt{\e_\La(2d\b + 2)}.
\end{equation}

For practical purposes (e.g., for using Monte Carlo methods to study the coverage probabilities, guaranteeing a high coverage probability against the randomness of the fading and the environment, etc.), it would also be of interest to work with a lattice $\La$ such that the coverage function near a base station (equivalently, it is logarithm) is the most \textit{stable}. This would amount to the choice of a lattice so as to minimize the fluctuations or the variance of $\log \c_0(\Phi)$. Once again, for a given nearest base station $\xi_\ze$, this amounts to choosing a lattice that minimizes the Epstein Zeta function $\e_\La(2d\b+2)$. The minimizing lattice for this in 2D is the triangular lattice, and in 3D, for our regime of interest $\b>1$, the minimizing lattice is conjectured to be the FCC.

We now examine the coverage probability itself, which, in view of the analysis presented above, would amount to considering the expectation of a log-normal random variable with parameters as specified in \eqref{eq:log-normal-parameters} (conditioned on $\xi_\ze$),
with $|\xi_\ze|$ following a $d$-dimensional standard Gaussian distribution. 

We first condition on $\xi_\ze$ and obtain the expectation of the log-normal as \[ \exp \left( - \t \s^{d\b} |\xi_\ze|^{d\b} \e_\La(d\b) + \frac{1}{2} \t^2 \b^2 \s^{2d\b + 2} |\xi_\ze|^{2d\b} {\e_\La(2d\b + 2)} \right). \] In the small $\s$-regime, we may approximate the exponent in this quantity by the leading term in $\s$. We can then take expectation with respect to $\xi_\ze$ following a $d$-dimensional standard Gaussian distribution to obtain the final coverage probability.

Thus, the coverage probability can be approximated, in the small $\s$ regime, by
\[ \E_{\xi_\ze \sim N(\ze,I_d)} \left[ \exp \left( - \t \s^{d\b} |\xi_\ze|^{d\b} \e_\La(d\b) \right) \right]. \]
We then observe that, if $\xi_\ze \sim N(\ze,I_d)$, then $|\xi_\ze|^2 \sim \chi^2_d$, that is, the Chi-squared distribution with $d$ degrees of freedom. The probability density function for the $\chi_d^2$ distribution on positive reals can be expressed as (\cite{johnson1995continuous})
\begin{equation}
f_{\chi^2_d}(x)=\frac{1}{2^{d/2}\Gamma(d/2)} x^{d/2 -
 1}e^{-x/2}, 
\label{eq:chi-squared}
\end{equation} 
where $\Gamma(\alpha)$ is the Gamma integral given by $\int_{0}^\infty x^{\alpha-1}e^{-x} \d x$.
We can therefore write the coverage probability above as
\begin{equation} \label{eq:covprob-formula}
 \frac{1}{2^{d/2}\Gamma(d/2)} \int_0^\infty u^{d/2-1} \exp\left( -\t \s^{d\b} \e_\La(d\b) u^{\frac{1}{2}d\b} - \frac{1}{2}u  \right) \d u.  
\end{equation}

\subsection{The regime of small $\s$ and large $\t$:  
  parametric dependencies of coverage probability} \label{sec:large-theta}

In this section, we explore the asymptotics (for large $\theta$) in the coverage probability near a base station, via the asymptotics of the integral \eqref{eq:covprob-formula}. The asymptotics of \eqref{eq:covprob-formula} below indicate a power law scaling as $\t^{-1/\b}$. In the context of Theorem \ref{thm:SINR} (which also exhibits a similar $\t^{-1/\b}$ scaling), the leading constant $C_2(\b,\s)$ in the asymptotics of \eqref{eq:covprob-formula} is analogous to the leading constant $C_1(\b,\s)$ in that theorem. The simpler algebraic structure of $C_2(\b,\s)$, however, unveils a clear power law dependency on $\s$, which suggests a similar $\s$-dependency for the limiting constant $C_1(\b,\s)$. In our considerations of this asymptotic, we consider $\s$ to be small but fixed, and $\t \to \infty$. It would be of interest to extend such suggestive  behaviour to obtain a theorem the rigorously proves a comprehensive dependency structure for $C_1(\b,\s)$; for reasons of brevity we leave that question for future research.
 


In the regime of large $\t$ and small $\s$, since $\b>1$, the integral \eqref{eq:covprob-formula} is approximately
\[ \frac{1}{2^{d/2}\Gamma(d/2)} \int_0^\infty u^{d/2-1} \exp\left( -\t \s^{d\b} \e_\La(d\b) u^{\frac{1}{2}d\b} \right) \d u.   \] Setting $v=u^{\frac{1}{2}d\b}$ and $a=\t \s^{d\b} \e_\La(d\b)$ in the above, we reduce the integral to
\begin{align}
\frac{1}{2^{d/2}\Gamma(d/2)}  \frac{2}{d\b} \int_0^\infty v^{\frac{1}{\b}-1}e^{-a v} \d v 
&= \frac{1}{2^{d/2-1}\Gamma(d/2)}  \frac{\Gamma(1/\b)}{d\b}a^{-1/\b} 
\nonumber \\
&= \frac{\Gamma(1/\b)}{d\b 2^{d/2-1}\Gamma(d/2) \e_\La(d\b)^{1/\b} \s^{d}} \cdot \t^{-1/\b} \nonumber \\
& = C_2(\b,\s)  \cdot \t^{-1/\b},  \label{eq:powerlaw}
\end{align}
where $$ C_2(\b,\s) = \frac{\Gamma(1/\b)}{d\b 2^{d/2-1}\Gamma(d/2) \e_\La(d\b)^{1/\b} } \cdot \s^{-d}.$$ 
Therefore, in the small $\s$ regime, \eqref{eq:powerlaw} recovers the large $\t$ asymptotics of the coverage probability near a base station 
as $\t^{-1/\b}$, and suggests a parametric dependency  of the limiting constant $C_1(\b,\s)$ (c.f. Theorem \ref{thm:SINR}) as an inverse power law in $\s$ as $ {\Gamma(1/\b)} \cdot \left({d\b 2^{d/2-1}\Gamma(d/2) \e_\La(d\b)^{1/\b} }\right)^{-1} \cdot \s^{-d} $. Such parametric dependence on $\s$ is corroborated empirically by Figs.~\ref{fig:2dim-SINR-1} -- \ref{fig:3dim-SINR-2}.


\subsection{The regime of small $\s$ and small $\t$ : asymptotic linearity} \label{sec:large-sigma} 

The regime of small $\t$ is important in the wireless network model for the following reason. We are able to detect a signal as soon as the SINR is above some threshold, that is, the SINR is not too low. From this perspective, it would be relevant to have $\P[\SINR > \t]$ to be high for small values of $\t$, with the pertinent question being its dependence on $\t$ as $\t \to 0$. Accordingly, we obtain an approximation of the coverage probability in the regime of small $\t$ and $\s$.

In the  integral \eqref{eq:covprob-formula}, we can approximate $ \exp\left( -\t \s^{d\b} \e_\La(d\b) u^{\frac{1}{2}d\b} \right)$ by \newline $\left(1 - \t \s^{d\b} \e_\La(d\b) u^{\frac{1}{2}d\b} \right)$ in the regime of small $\t$, and therefore obtain an approximation for the coverage probability near a base station as
\begin{align*}
\lefteqn{1 - \t \s^{d\b} \e_\La(d\b) \frac{1}{2^{d/2}\Gamma(d/2)} \int_0^\infty u^{d/2-1}u^{\frac{1}{2}d\b} e^{-\frac{1}{2}u} \d u} \\
&= 1 - \t \s^{d\b} \e_\La(d\b) \frac{2^{\frac{1}{2}d(\b+1)}\Gamma(\frac{1}{2}d(\b+1))}{2^{d/2}\Gamma(d/2)} \\
&= 1 - 2^{\frac{1}{2}d\b} \s^{d\b} \e_\La(d\b) \frac{\Gamma(\frac{1}{2}d(\b+1))}{\Gamma(\frac{1}{2}d)} \cdot \t.  \numberthis \label{eq:small-theta}  \\
 \end{align*}
Thus, in the small $\t$ regime, the coverage probability near a base station decays approximately linearly in $\t$, with the slope being given by \eqref{eq:small-theta}. We once again observe a worsening behaviour of the coverage probability with increasing $\s$.


\subsection{Monotonicity in $\s$} \label{sec:monotonicity}

 In this section, we will demonstrate that, in the regime of small $\s$, the coverage probability near a base station for Gaussian perturbed lattice networks, for a given threshold $\t$, is monotonically decreasing in the dispersion $\s$ of the perturbation.
 
 To this end, we will consider the coverage function near a base station $\c_0(\Phi,\s)$ : 
 \begin{align*}
 \log \c_0(\Phi,\s) = & - \sum_{\la \in \La_\ze} \log \left(1 + \theta \frac{|\s \xi_\ze|^{d \b}}{|\la + \s \xi_\la|^{d \b}} \right) \\
  = & - \sum_{\la \in \La_\ze} \log \left(1 + \theta |\xi_\ze|^{d\b}\frac{ 1}{ \left|\frac{1}{\s}\la +  \xi_\la\right|^{d \b}} \right).  
  \end{align*}
The derivative of $\log \c_0(\Phi,\s)$ with respect to $\s$ is given by 
 \begin{align*}
\lefteqn{\frac{\partial}{\partial \s} \log \c_0(\Phi,\s)} \\ 
  &= -\sum_{\la \in \La_\ze} \left(1 + \theta |\xi_\ze|^{d \b}\frac{ 1}{ \left|\frac{1}{\s}\la +  \xi_\la\right|^{d \b}} \right)^{-1} 
  \frac{- \b \t | \xi_\ze|^{d \b}}{\left|\frac{1}{\s}\la +  \xi_\la\right|^{d \b + 2}} \frac{\partial}{\partial \s} \left[ \left|\frac{1}{\s}\la +  \xi_\la \right|^{2}  \right] 
\numberthis \label{eq:monotonic}
 \end{align*}
and 
 \begin{align*} 
 \frac{\partial}{\partial \s} \left[ \left|\frac{1}{\s}\la +  \xi_\la \right|^{2} \right] 
&= 2 \langle \frac{1}{\s}\la +  \xi_\la ,  \frac{-1}{\s^2}\la \rangle \\
&= -2 \left( \frac{|\la|^2}{\s^3} + \frac{\langle \xi_\la ,  \la \rangle }{\s^2} \right) \\
&= -2 \frac{|\la|^2}{\s^3}  \left(  1 + \s \cdot \frac{ \langle \omega_\la, \xi_\la \rangle}{|\la|}  \right), 
  \end{align*}
  where $\omega_\la \in \mathbb{S}^{d-1}$ in the direction of the vector $\la \in \R^d$, as before.
  
In the small $\s$ regime, with high probability  $\left|\s \cdot \frac{ \langle \omega_\la, \xi_\la \rangle}{|\la|} \right| \ll 1$ for all non-zero $\la \in \La_\ze$, which, in light of \eqref{eq:monotonic}, implies that the logarithmic 
derivative $\frac{\partial}{\partial \s} \log \c_0(\Phi,\s)$ is negative.

Thus, in the small $\s$ regime, the coverage function near a base station is, with high probability, monotonically decreasing in $\s$. Since the coverage probability is the expectation of the coverage function, this indicates in the small $\s$ regime, the coverage probability near a base station would decrease with increasing $\s$.

 \section{Choice of lattice and energy optimality}
\label{sec:energy-optimality}
An important question that arises in studying disordered lattices as models for wireless base stations is the choice of the lattice which we perturb. In this direction, we are guided by considerations of energy optimality of lattices, which appears rather interestingly in our investigations of the coverage probability.

To be more precise, we can consider Gaussian perturbations of a lattice $\Lambda$ with dispersion parameter $\sigma$. The coverage probability curve is then a function of $\sigma$. Although this function is not analytically tractable, we can nevertheless expand it in a series in $\sigma$ in the regime where the parameter $\sigma$ is small, as we demonstrate in Section \ref{sec:small-sigma}. The coefficients of this expansion, naturally, are functionals of the lattice $\Lambda$. 

 It turns out, as in \eqref{eq:logcov-Epzeta}, that the coefficient of the leading term in this expansion is the celebrated Epstein Zeta Function of the lattice $\Lambda$, which, heuristically speaking, can be thought of as a lattice energy, and  has deep connections to sphere packing, number theory, crystallography, quantum field theory and other diverse areas of mathematics and physics (see, e.g.,\cite{titchmarsh1986theory,terras2012harmonic,elizalde2012ten,chowla1949epstein}).  Maximizing the coverage probability at a given level $\theta$ would amount to considering the lattice that minimizes the Epstein Zeta Function,  a classical problem in its own right, that has connections to other fundamental questions like the crystallisation conjecture (\cite{SS,blanc2015crystallization,dettmann2012new,betermin-lecture,serfaty2015coulomb,sandier20152d,grabner2014point}). In 2D Euclidean space, the minimizing lattice for the Epstein Zeta function is the triangular lattice, which is the focus of our attention. In 3D Euclidean space, a rigorous understanding of minimal lattices for the Epstein Zeta function is limited, but, as we argue in Section \ref{sec:3D}, a natural choice to focus on is the Face Centered Cubic (FCC) lattice and its Gaussian perturbation (i.e., the perturbed FCC (abbrv. PFCC)).

\section{Comparison of random point sets and persistent homology}
\subsection{Comparison via persistent diagrams}
\label{sec:PD}
One of the issues that we address in our investigations is
the comparison between point processes that are candidates
for modelling the locations of the wireless base
stations. An interesting question on its own right, this is
also motivated by the desire to find disordered lattices
which are appropriate ``substitutes'' for the random matrix
networks - in particular, this entails a comparison between
the two point processes. To compare two point
configurations, we appeal to their \textit{persistence
diagrams} {(abbrv. PDs)}, a tool that has  recently attracted a lot of
interest in topological data analysis. For details on
persistence diagrams and random topology, we refer the reader to recent articles (\cite{edelsbrunner2008persistent,weinberger2011persistent, yogeshwaran2015topology, HST18}), to provide a partial list. In the limited scope available to us here, we give a succinct heuristic description as follows, with concrete definitions in Appendix~\ref{a:PD}. 

We can consider random geometric graphs on a point set by
connecting pairs of points by an edge when they are nearer
than a given threshold $\eps$, and more generally form a
$k$-clique out of $k$ points if the $\eps$ balls around
these $k$ points intersect pairwise. As $\eps$ varies from
$0$ to $\infty$, more points get connected with each other. 
The topological properties of the point set can be discussed in terms of  \textit{homology groups} related to certain \textit{complexes} arising out of this construction.

Heuristically speaking, such considerations entail that the fundamental topological properties  of the point set are captured by certain
\textit{holes} embedded in the point set fattened by the $\eps$-balls.
 Holes appear and disappear (due to the overlap of the
balls around the points) as the connectivity threshold
$\eps$ changes. The most significant ones among these
holes  are those that
\textit{persist} for a long time, that is, the thresholds for
their appearance and disappearance are well-separated. Morally, such \textit{persistent} holes reflect a fundamental feature of the
point set compared to less persistent ones, whose appearance could be attributed to random noise.
The PD corresponding to the point set is a 2D
plot against each other of these two thresholds (resp. for
appearance and disappearance) for these holes. 

We compare two point configurations by comparing their
persistence diagrams, which brings us to the natural
question of comparing two persistence diagrams. To this end, we 
adopt two approaches.

For the first approach, we consider the persistence diagrams as atomic probability measures. We then proceed to compute the \textit{Total Variation 
distance} (abbrv. TV distance) between these two measures. For two probability measures having densities $f$ and $g$ on the same Euclidean space, it can be expressed simply as the integral $\int |f(x) - g(x)| \d x$. For computational simplicity, we convolve the atomic measures given by the PDs with Gaussians having a small dispersion (equal to $1/2$) and discount the contribution from atoms near diagonal of the PD (cf. \cite{KHF16}), and take the TV distance between the resulting measures with densities via the above formula.

For the second approach, we make use of the fact that PDs themselves are 2D finite
point sets, and we compute the \textit{distance} between two
PDs by computing their  \textit{symmetrized
nearest-point distance}  (\cite{mateu2010distances}). ln particular,
let $X,Y$ be two finite point configurations in a metric
space (equipped with the metric $d$), and for $x \in X$, let
$y(x)$ be the nearest point to $x$ in $Y$ (making arbitrary
choices to break ties, if any). Then define a distance
between the configurations $X,Y$ as
$\overline{d}(X,Y)=\sum_{x \in X} d(x,y(x))$. Now
$\overline{d}(\cdot,\cdot)$ is clearly asymmetric in its
arguments, so we define the symmetrized nearest-point
distance between $X$ and $Y$ as $D(X,Y)=\overline{d}(X,Y) +
\overline{d}(Y,X)$.

We mention in passing that other metrics for measuring distance between PDs have been considered, e.g. the 
\textit{bottleneck distance}  (\cite{cohen2007stability}). However, the calculation of such metrics on given data sets can often be highly computationally intensive. In this article, we focus on the total variation and symmetrized nearest-point distances for their simplicity and computational tractability. The comparative study of wireless network distributions (and more generally, point processes) with  respect to other metrics on PDs certainly warrant further investigation;
 though, as exemplified by the  consistency of the minimality threshold of around $\sigma=0.4$ across our chosen metrics in this paper,  we expect our broad conclusions to be more or less robust to the specifics of the metrics involved.

\begin{figure}[htbp]
\begin{minipage}{0.5\hsize}
\begin{center}
\includegraphics[scale=0.55]{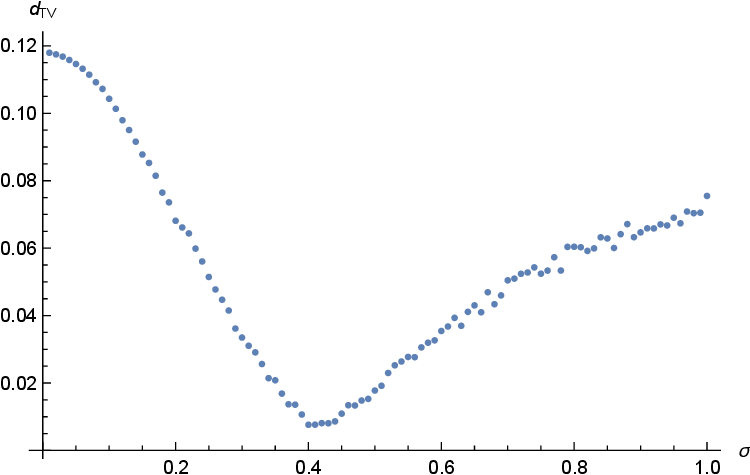}
\end{center}
\end{minipage}
\begin{minipage}{0.5\hsize}
\begin{center}
\includegraphics[scale=0.55]{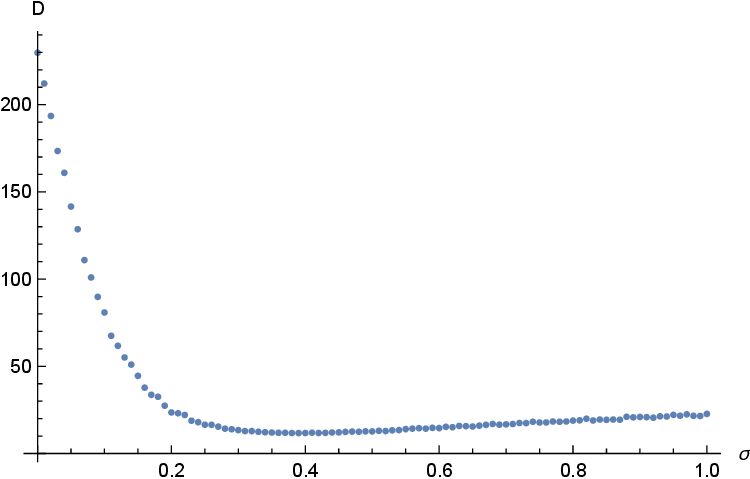}
\end{center}
\end{minipage}
\caption{Distance between Persistence Diagrams of
 Ginibre and PTL  point processes. Left panel: Total Variation distance between the PDs considered as point measures. Right panel:  $D(\cdot,\cdot)$ distance between the PDs. 
Around $\sigma=0.4$, the PDs are
 the closest in both metrics. }
 \label{fig:PD}
\end{figure}

It may be noted that PDs of random point sets are random 2D point sets themselves.  So, eventually we are comparing another pair of (random) point sets rather than the original point processes directly. 
However, the comparison of PDs enjoys two advantages over a direct comparison of the point processes. 

First, the PDs are always 2D point sets located on the same domain (the positive 
quadrant in $\R^2$). As such, their comparison via the metric $\overline{d}$ is always reasonable, irrespective of the nature of the ambient spaces in which the original point processes are embedded (e.g., this allows for the comparison of point processes that are embedded in Euclidean spaces of different dimensions). Thus, the comparison of PDs truly focuses on a desirable comparison of the structural properties of the 
point set and  is oblivious to extraneous factors like its physical embedding. 

Secondly, passing from the original point process to the PD controls the influence of outliers, and consequently has a stabilizing effect with respect to random noise, which is better for comparison purposes.

\subsection{Comparison via nearest neighbour distribution}
\label{sec:nnd}
An important aspect of point patterns is their nearest
neighbour distribution (abbrv. nnd), that is, the
statistical distribution of the typical point from its
nearest neighbour. A renowned instance of this is the 1D
case, where this reduces to the famous level spacing
distribution that has been widely investigated in the
context of random matrix theory (\cite{mehta2004random}). A
significant result in this direction, that traces its origin
as far back as Wigner's work, is that the level spacing
distributions of Gaussian random matrix ensembles are very
different from, say, the \textit{independent case} (i.e.,
the Poisson point process), and a great volume of research
has been dedicated to successfully establishing the
conjecture that such behaviour is, in fact,
\textit{universal} (i.e., not dependent on the Gaussianity
or other specifics of the matrix distributions)  (see, e.g.,
\cite{erdos2012universality,tao2011random}).  In particular,
the repulsion among the points in the Ginibre ensembles is captured by the fact that the \textit{level spacing }(or \textit{gap}) distributions (the so-called Wigner distributions) have a vanishing density near the origin, whereas for independent points, the gaps follows an exponential  distribution (which, in particular, has its mode at the the origin). While  persistence diagrams certainly capture  point sets in a much more comprehensive manner, the nnd-s are much simpler and succinct summaries (albeit more limited), and can be considered as an alternative possibility for comparing point sets when computational load is  a bigger consideration than high accuracy.

We  examine nnd-s in 2D and 3D, which are much less
understood than the 1D case. While the nnd for a perturbed
lattice model can be expressed, in principle, as an infinite
series in terms of various Gaussian probabilities, in
practice such expressions are of little utility as they do
not shed much light on the statistical or analytical
properties of the relevant distribution. In this article, we
undertake an empirical investigation of the nnd-s for
perturbed lattice models, comparing them against their
counterparts for the Ginibre and the
Poisson models, as relevant. Rigorous analytical exploration
of their distributional properties, for instance in
comparison to the Poisson and the Ginibre
models, would be a natural avenue for future research that
appears to be beyond the reach of current methods. We
observe in passing that, as $\sigma$ increases, the nnd-s
for perturbed lattice models converge to that of the
Poisson, thereby corroborating the overall convergence at
the level of point processes. 

We do not extensively use nnd-s as a metric for comparison of point sets in our
investigations of spatial network models in this paper,
we defer the results of the empirical investigations on nnd-s
to Appendix~\ref{a:nnd}. 


\section{Disordered lattices for optimal SINR in 2D and 3D}
\label{sec:DisLat_comp}

\subsection{2D planar networks} \label{sec:2D}
The lattice which minimizes the Epstein zeta function in 2D is the triangular lattice, where the fundamental lattice domain is in the shape of a rhombus of unit sidelength in $\R^2$ (\cite{rankin1953minimum,diananda1964notes,cassels1959problem,ennola1964lemma,ennola1964problem}). This is the same as the famous Abrikosov lattice that plays an important role in statistical physics theoretical physics, e.g. through its emergence as the ground state in the celebrated theory of Ginzburg-Landau vortices and  Coulomb gases (a.k.a. the 2D one component plasma) (\cite{abrikosov1957magnetic,sandier2008vortices,serfaty2015coulomb,sandier20152d}).

For our purposes, we consider Gaussian perturbations of the triangular lattice, with the lattice spacing scaled  so as to have on the average one point per unit area. We will refer to this point process as the \textit{Perturbed Triangular Lattice} (abbrv. PTL). We study the coverage probability of these perturbed lattices indexed by the dispersion $\sigma$ of the Gaussian perturbations, and plot the coverage probability against the corresponding threshold $\theta$. 

The coverage probabilities are computed via Monte Carlo simulations, generating a large number of realizations of the random point configurations, computing the corresponding SINR, and obtain the coverage probability from the histogram of SINRs. To be precise, we generate 20,000 samples for each point process (PTL, Ginibre, Poisson) to compute the mean of SINR as a function of $\theta$. 
Thus, we obtain a family of curves plotting $p_c(\theta)$ against $\theta$, the curves being indexed by the parameter $\sigma$. The results are exhibited in Fig.~\ref{fig:2dim-SINR-1} and Fig.~\ref{fig:2dim-SINR-2}.

These plots exhibit several interesting features. We empirically observe a strict monotonicity in $\sigma$: for $\sigma=0$, i.e. no disorder, the coverage probability $p_c(\theta)$ is the highest for a given threshold $\theta$, and $p_c(\theta)$ decreases monotonically as $\sigma$ increases, always staying above the corresponding curve for the Poisson distribution, but approaching it as $\sigma \to \infty$. 

One of our goals is to obtain a disordered lattice model
which can substitute for the random matrix network, both in
terms of similarity as point configurations as well as in
terms of the behaviour of the SINR. Heuristically, for
$\sigma$ near 0 the model would strongly resemble the
original lattice, whereas for $\sigma$ large, Poissonian
behaviour sets in. In particular, when $\sigma$ is too
small, we can trace most points back to the lattice site
where it came from, whereas for $\sigma$ too large, the
\textit{memory} of an ordered structure is completely
lost. It is natural, therefore, to look for random matrix
behaviour somewhere in between, away from this extremities
for the disorder parameter.  A rule of thumb, therefore,
might be to look at disorder of magnitude  like half of the
lattice spacing, so that the identification of the parent
sites for the perturbations of neighbouring lattice points becomes only just improbable.

As discussed in Section \ref{sec:PD}, we compare lattice
perturbations with the Ginibre network by comparison of their
PDs, which in turn is achieved by computing two alternate
metrics. The first metric is their TV distance as atomic
probability measures (calculated after smoothing by a
localised Gaussian kernel). The second metric is  the
symmetrized nearest point distance between them. These
measures, once again, are computed via Monte Carlo
simulations of the point processes, using 100 samples for
each value of $\sigma$.  As unveiled in Fig.~\ref{fig:PD},
it turns out that in both metrics, the closest approximation
of the Ginibre network (i.e., the Ginibre network) by a Gaussian Perturbed Triangular Lattice is achieved around $\sigma=0.4$. In this vein, we make particular note of the relatively sharp convexity of the TV distance curve near its minimum around $\sigma=0.4$.

The nearest neighbour distributions (abbrv. nnd-s) of the three point processes are displayed in Fig.~\ref{fig:2dim-nearest-neighbour}, for the Poisson, Ginibre and PTL for various values of $\s$. The plots for the nnd-s are generated empirically via Monte Carlo simulations, using 10,000 realizations of the relevant point process for each curve. It may be observed that, as $\s$ increases, the nnd of the corresponding PTL converges to that of the Poisson point process. The closest the nnd for a PTL gets to the nnd of the Ginibre point process is around $\s=0.4$. This is also the perturbation value around which the distance between the corresponding PDs is minimized, and the SINR vs threshold curves nearly overlap.

It may be observed that for $\sigma = 0.4$, the coverage
probability curve of the disordered lattice lies close to,
and in fact, slightly above the corresponding curve for the
Ginibre network (see Fig.~\ref{fig:2dim-SINR-1} and
Fig.~\ref{fig:2dim-SINR-2}). This indicates that a wireless
network based on a disordered triangular lattice with
disorder $\sigma$ around 0.4 can be taken as an effective
replacement for the Ginibre network, in terms of coverage
probability for the PTL performing at least as well as the
Ginibre. We believe that this can have significant impact on the design and investigation of random wireless networks, particularly from the point of view of large scale numerical and computational research.

\def\size{0.53}
\begin{figure}[htbp]
\begin{minipage}{0.52\hsize}
\begin{center}
\includegraphics[scale=\size]{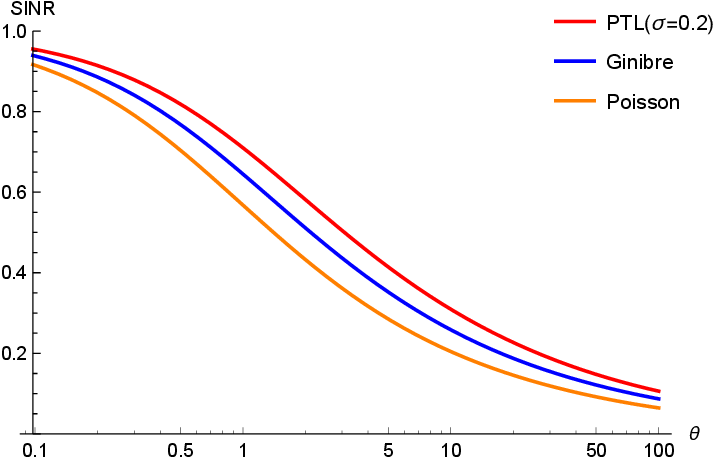}
\end{center}
\end{minipage}
\begin{minipage}{0.48\hsize}
\begin{center}
\includegraphics[scale=\size]{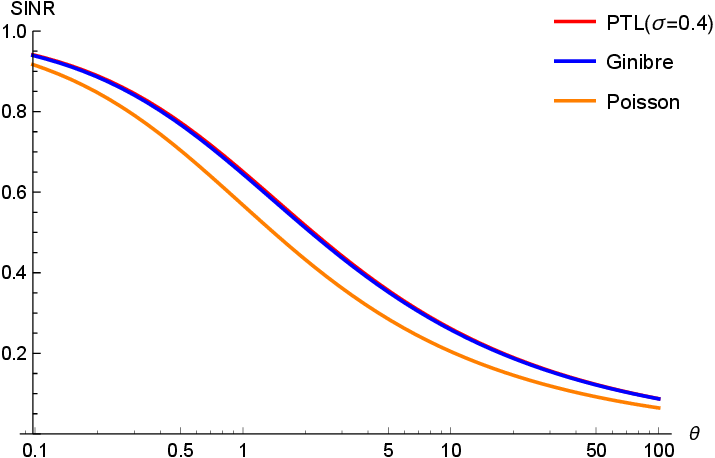}
\end{center}
\end{minipage}\\[4mm]
\begin{minipage}{0.52\hsize}
\begin{center}
\includegraphics[scale=\size]{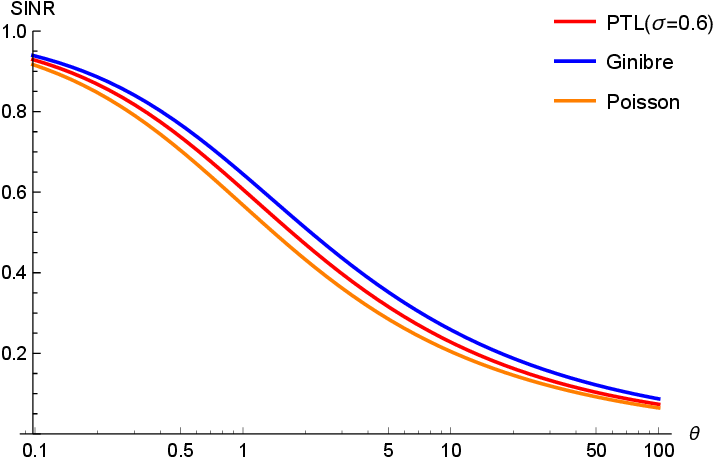}
\end{center}
\end{minipage}
\begin{minipage}{0.48\hsize}
\begin{center}
\includegraphics[scale=\size]{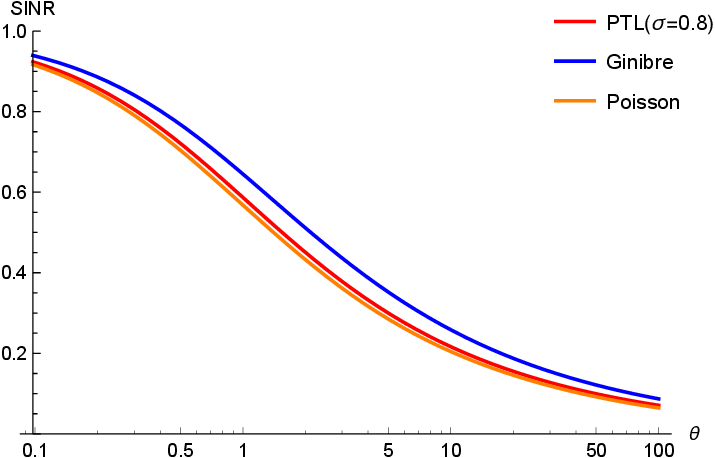}
\end{center}
\end{minipage}
\caption{Coverage probability vs SINR threshold curves for Perturbed Triangular Lattice(PTL), Ginibre and
 Poisson ($\sigma=0.2, 0.4, 0.6, 0.8$) for $\beta=2$. }
\label{fig:2dim-SINR-1}
\end{figure}
 \begin{figure}[htbp]
 \begin{center}
 \includegraphics[scale=0.6]{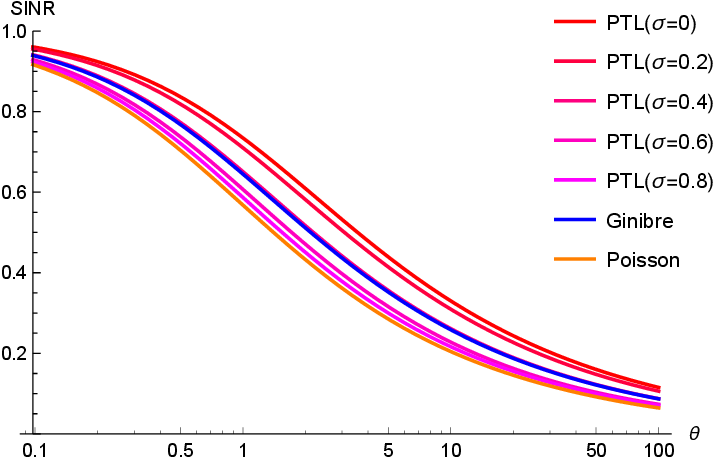}
 \end{center}
 \caption{Coverage probability vs SINR threshold curves for Perturbed Triangular Lattice (PTL) ($\sigma=0, 0.2, 0.4, 0.6, 0.8$),
  Ginibre, Poisson for $\beta=2$. 
Coverage probability curves for PTL decrease as $\sigma$ increases, match Ginibre, and tend towards that of Poisson. }
\label{fig:2dim-SINR-2}
 \end{figure}

\subsection{3D spatial networks} \label{sec:3D}

The ordinary cubic lattice has a cube with eight lattice points (i.e., $[0,a]^3$ with a lattice point on each corner) as a unit cell, which is denoted by $\mathcal{C}$, and it is formed as $\bigcup_{x \in a\Z} (\mathcal{C} + x)$.
The unit cell of the Body Centered Cubic (abbrev. BCC)
lattice is $\mathcal{C}$ with one more lattice point
$(a/2,a/2,a/2)$ in the center and
that of	Face Centered Cubic (abbrev. FCC) lattice is
$\mathcal{C}$ with a lattice point on each face (see (cf. \cite[Example 8.3]{S13}) for BCC and FCC lattices). 

In 3D Euclidean space, identification of the base lattice to disorder poses a challenge, stemming from the fact that the minimizing lattice for the Epstein zeta function is not fully understood in 3D (\cite{SS,betermin-lecture,betermin2019local}). 
It was shown by Ennola that the FCC lattice as well as the BCC lattice are local minimizers of the Epstein
Zeta function in the space of lattices. 
However, there is no definitive understanding of what the global minimizer is for a given value of $s$. It was conjectured by Sarnak and Strombergsson that, for $s>3$ (cf. formula \eqref{eq:epzeta}), the minimizing lattice for the Epstein Zeta function $\e_\L(s)$ in 3D Euclidean space is the FCC (\cite{SS,betermin-lecture,betermin2019local}). In our study, therefore, we will henceforth be using the disordered FCC lattice, leaving the issue of a completely rigorous optimal choice of lattice to future breakthroughs in the theory of lattices.

In our investigations, we consider Gaussian perturbations of the FCC lattice and the cubic lattice (i.e., $\mathbb{Z}^3$), for various values of the standard deviation $\s$, and compare them against the Poisson point process of the same intensity in 3D Euclidean space. We plot the coverage probability vs thresholds for networks given by these processes at various values of $\s$. The plots are generated empirically via Monte Carlo simulations, using 10,000 realizations of the relevant point process for each curve.  The results are displayed in Fig.~\ref{fig:3dim-SINR-1} and \ref{fig:3dim-SINR-2}. The plots show the perturbed FCC (abbrv. PFCC) network to be the clearly the best performer for small values of $\s$, and it is $p_c(\t)-\t$ curve tends to match that of the Perturbed Cubic Lattice (abbrv. PCL) network for larger values of $\s$ starting from around $\s=0.5$, or half the lattice spacing.  For all values of $\s$, however, the PFCC seems to be performing at least as well as the PCL, suggesting that it would be the better choice as the base lattice to perturb. 

 In Fig.~\ref{fig:3dim-SINR-2}, we focus on the behaviour of the $p_c(\t)$-$\t$ curve of the PFCC lattice for varying $\s$, and compare them with the corresponding curve for the Poisson network. We observe worsening performance of the networks with increasing values of $\s$, reflected in the fact that the curves keep  getting pushed down. The Poisson network, for its part, appears to be performing  worse compared to the PFCC (as well as the PCL) uniformly for all values of $\s$, and it is only for large $\s$ that the coverage probability vs threshold curve of the perturbed lattice models tend to converge to that of the Poisson model. 

We complement our study by an investigation of the nearest
neighbour spacing distributions of these point process in
$\R^3$. The plots for the nearest neighbour distributions
are generated empirically via Monte Carlo simulations, using
10,000 realizations of the relevant point process for each
curve. The results are displayed in
Fig.~\ref{fig:3dim-nearest-neighbour}. It may be observed
that the nnd for the PFCC process is in general a bit more
concentrated around its mode than the PCL or the Poisson,
indicating a more homogeneous distribution of points in
space. The Poisson point process, on the other hand, exhibits a very \textit{flat} nnd profile, indicating both clumps of points and large holes. It may also be observed that the nnd curves for the PFCC and the PCL converge around $\s=0.5$, which is also the value of the perturbation at which the coverage probability curves coincide. Finally, for large values of $\s$ the nnd curves for  both the perturbed lattice models converge to that of the Poisson, reflecting distributional convergence of the underlying point processes.

\def\size2{0.5}
\begin{figure}[htbp]
\begin{minipage}{0.52\hsize}
\begin{center}
\includegraphics[scale=\size2]{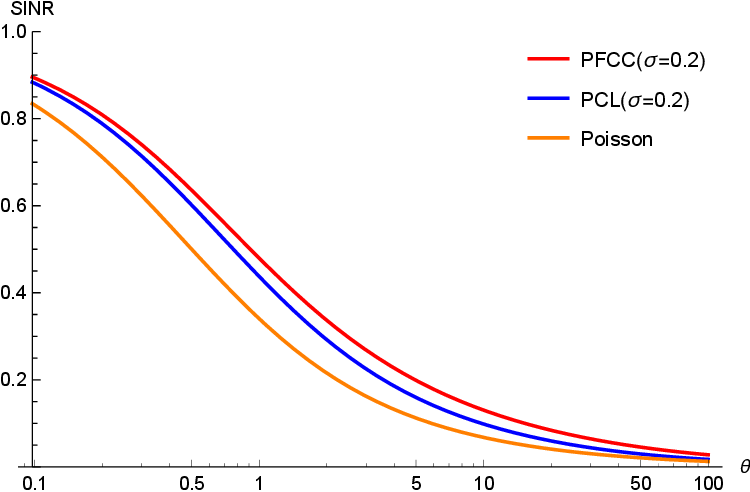}
\end{center}
\end{minipage}
\begin{minipage}{0.48\hsize}
\begin{center}
\includegraphics[scale=\size2]{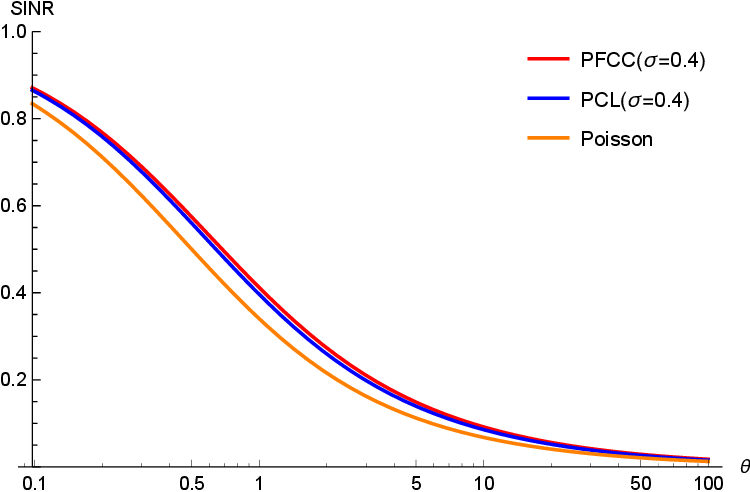}
\end{center}
\end{minipage}\\[4mm]
\begin{minipage}{0.52\hsize}
\begin{center}
\includegraphics[scale=\size2]{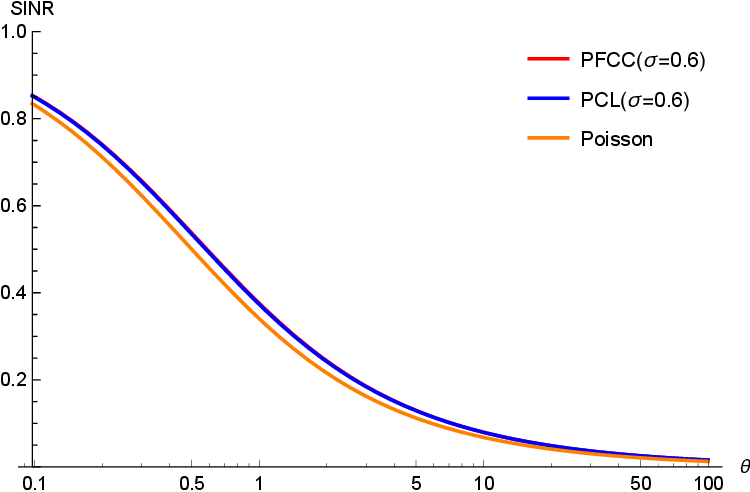}
\end{center}
\end{minipage}
\begin{minipage}{0.48\hsize}
\begin{center}
\includegraphics[scale=\size2]{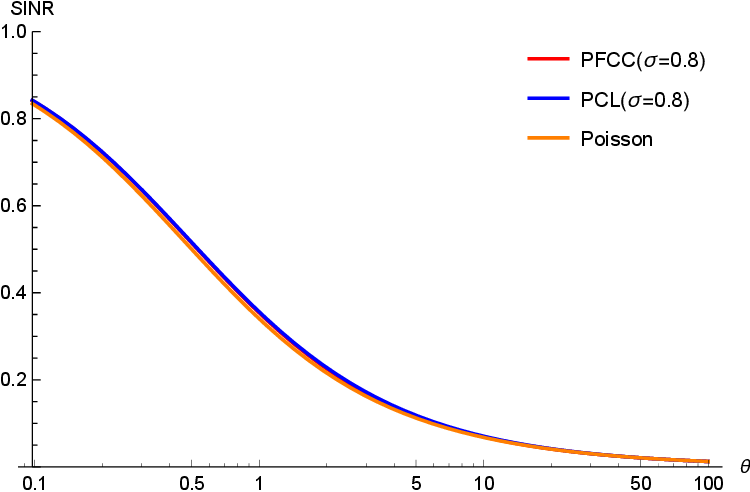}
\end{center}
\end{minipage}
 \caption{Coverage probability vs SINR threshold curves for
 Perturbed FCC (PFCC), Perturbed Cubic Lattice (PCL)
 ($\sigma=0.2, 0.4, 0.6, 0.8$), Poisson for $\beta=4/3$. }
\label{fig:3dim-SINR-1}
\end{figure}
\begin{figure}[htbp]
\begin{center}
\includegraphics[scale=0.65]{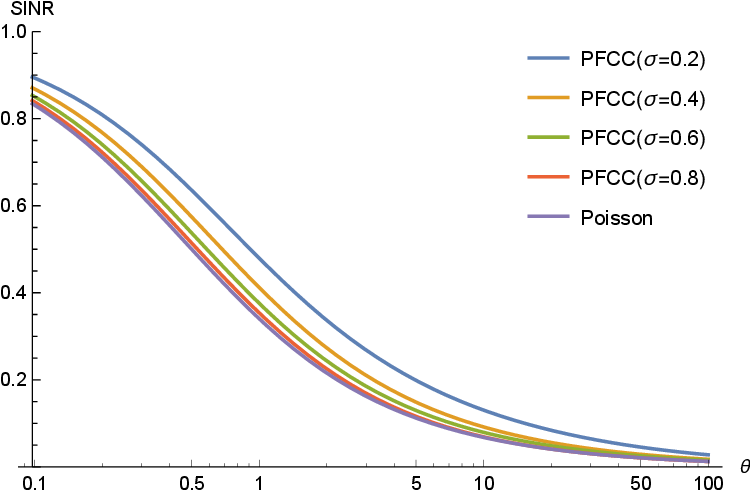}
\end{center}
\caption{Coverage probability vs SINR threshold curves for Perturbed FCC lattice (PFCC) ($\sigma=0.2, 0.4, 0.6, 0.8$) from
 above and Poisson for $\beta=4/3$. Coverage probability curves for PFCC decrease as $\sigma$ increases and tend towards that of Poisson. }
\label{fig:3dim-SINR-2}
\end{figure}

\newpage
\section{Interpolation with Poisson in the high noise regime} \label{sec:conv-Poisson}
As we have seen in \ref{sec:2D} and \ref{sec:3D} for SINR and also 
Fig.~\ref{fig:2dim-nearest-neighbour} and
Fig.~\ref{fig:3dim-nearest-neighbour} in Appendix~\ref{a:nnd} for nnd-s of point processes, when the dispersion $\sigma$ gets large, its SINR and nnd become close to those of Poisson. 
In this section, we will rigorously demonstrate the convergence of the coverage probabilities of the Gaussian perturbed lattice networks to that of the Poisson network for a given threshold $\t$, for any starting lattice $\L$. 
 In fact, we will establish a more general result in Theorem \ref{result:conv}. To this end, we introduce the following notation. For any point process $\p$ on $\R^d$, define the quantity  
 \[\Sigma_R^\b(\p):=\sum_{\mathbf{\zeta} \in \p : |\mathbf{\zeta}|>R } \frac{1}{|\mathbf{\zeta}|^{d \b}}. \]
 We are now ready to state
 \begin{thm} \label{result:conv}
 Let $\{\Xi_\s\}_{\s>0}$ be a collection of point processes on $\R^d$ converging in distribution to a point process $\Xi$ on $\R^d$ as $\s \to \infty$, such that the random variables 
 $\Sigma_R^\b(\Xi_\s) \to 0$ in probability 
 as $R \to \infty$, uniformly over the collection $\{\Xi, \{\Xi_\s\}_{\s>0}\}$. Then, for any $\t>0$, the coverage probabilities $p_c(\t,\Xi_\s)$ for $\Xi_\s$ converge, as $\s \to \infty$, to the coverage probability 
 $p_c(\t,\Xi)$ for the point process $\Xi$.
  \end{thm}
We defer the proof of Theorem \ref{result:conv} to Section \ref{sec:proofs}. In what follows, we will lay down the argument to demonstrate the convergence to Poisson for perturbed lattice networks in the regime of high noise, using Theorem \ref{result:conv}.
  
In order to deduce the convergence of the coverage
probabilities of Gaussian lattice perturbations $\Xi_\s$
(with variance $\s$) to that of the Poisson network $\s$, we
first observe that as $\s \to \infty$, the point processes
$\Xi_\s \to \Xi$ in distribution, in the vague topology on
the space of locally finite point confuigurations on
$\R^d$. To see this via a simple dynamical argument, we
observe that an infinite system of particles on $\R^d$,
starting at the nodes of a lattice and evolving via
non-interacting Brownian motions converges in the long time
limit to the homogeneous Poisson point process on $\R^d$ having the same particle density as the initial configuration. For details, we refer the interested reader to \cite{stone1968theorem}, in particular to Theorem 2 and Example 2 therein. Now, if we start this dynamics from the initial configuration having one particle each at the lattice sites of $\L$, then after time $\s$, the point configuration evolves to have the same distribution as the perturbed lattice model at disorder $\s$. Thus, by virtue of this coupling we observe that, in the limit $\s \to \infty$, the perturbed lattice model converges to the homogeneous Poisson point process.

It, therefore, suffices to establish that the uniform convergence to zero of the tail sums \[\Sigma_R^\b(\Xi_\s)=\sum_{\mathbf{\zeta} \in \Xi_\s : |\mathbf{\zeta}|>R } \frac{1}{|\mathbf{\zeta}|^{d \b}}.\] To this end, we may compute the expectation $\E[\Sigma_R^\b(\p)]$ for a translation invariant point process $\p$ on $\R^d$ with unit intensity (i.e., density of points) with respect to Lebesgue measure, and obtain \[\E[\Sigma_R^\b(\p)] = \int_{ \{u \in \R^d:|u|>R\} }  \frac{1}{|u|^{d\b}}\d u.\] The last expression clearly converges to $0$ as $R \to \infty$ at a rate that does not depend on the specific point process $\p$ under consideration (thereby leading to uniform convergence in $\s$ in our setting). This establishes the uniform convergence  $\E[\Sigma_R^\b(\Xi_\s)] \to 0$. By Markov's inequality, this implies the desired uniform convergence in probability.

\section{Concluding remarks} \label{sec:higherD}

The problems studied in this work can, as easily, be posed for dimensions $d \ge 3$. In this work, we limit ourselves mostly to dimensions 2 and 3, motivated principally by their relevance to most commonly studied spatial network models. However, perturbed lattices and their network properties do pose an intriguing mathematical challenge in dimensions $d \ge 3$, particularly through the natural connections (in analogy to 2D and 3D) to energy-minimizing lattices. This brings to mind, for example, the case of dimension $24$, where the famous Leech Lattice has already been shown to exhibit many remarkable properties - both from the point of view of pure mathematics and also in applications to coding theory and analog-to-digital conversion (\cite{borcherds1985leech,cohn2017sphere,conway1986soft,conway2013sphere}). Energy minimizing phenomena for lattices are also known to be relatively well-understood in dimensions 4 and 8 (in addition to dimension 24)  (\cite{SS,betermin-lecture,betermin2019local}).

\section{Proofs of Theorems \ref{thm:SINR} and \ref{result:conv}} \label{sec:proofs}

\subsection{Proof of Theorem \ref{thm:SINR}} \label{sec:thmproof}

We begin by recalling the notation that
\[
	X_n(\sigma) = n + \sigma \gamma_n, (n \in \L \subset \R^d, \sigma > 0).
\]
Let
\[
	Y_n(\sigma):= |\sigma^{-1} X_n(\sigma)|^2 = \left |\frac{n}{\sigma} + \gamma_n \right|^2.
\]
Then $\{Y_n(\sigma)\}$ are independent random variables distributed as follows.
\begin{lem}\label{lem:distribution-of-Y}
$Y_n(\sigma)$ has the following probability density function
\[
	f(t,n, \sigma) = e^{-\frac{|n|^2}{2\sigma^2}}  \cdot t^{\frac{d}{2}-1}e^{-\frac{1}{2}t} \cI_d(\sigma^{-1}|n|\sqrt{t}),\quad (t\ge 0),
\]
where $\cI_d(u)=\frac{1}{2\cdot (2\pi)^{d/2}}
 \int_{\mathbb{S}^{d-1}} e^{u \langle \omega , e_1
 \rangle} \d \omega$, with $e_1$ being the first standard co-ordinate vector in $\R^d$ and $\d \omega$ being the standard spherical measure on $\mathbb{S}^{d-1}$.
\end{lem}

\begin{proof}[Proof of Lemma \ref{lem:distribution-of-Y}]
For a measurable function $F$, use a change of variables, we can deduce the following
\begin{align*}
\E[F(Y_n(\sigma))] &=\E\bigg[F \Big(\Big|\frac{n}{\sigma} + \gamma_n \Big|^2\Big)\bigg] \\
&= \int_{\R^d} F(|z|^2) \cdot \frac{1}{(2\pi)^{d/2}} e^{-\frac{1}{2}|z-\frac n\sigma|^2} \d z \\
&= \int_0^{\infty} \bigg(  \frac{1}{(2\pi)^{d/2}} \int_{\mathbb{S}^{d-1}} F(r^2)
 e^{- \frac{1}{2}(r^2 + \frac{|n|^2}{\sigma^2} - 2r\frac{|n|}{\sigma}\langle \omega, e_1 \rangle)} \d \omega \bigg) r^{d-1} \d r \\
&= \int_0^{\infty}  F(r^2)  r^{d-2}e^{-\frac{1}{2}\left(r^2 + \frac{|n|^2}{\sigma^2}\right)} \cI_d(\sigma^{-1}|n|r) \cdot 2r \d r \\
&=  e^{-\frac{|n|^2}{2\sigma^2}}
\int_0^{\infty} F(t)  t^{\frac{d}{2}-1}e^{- \frac{1}{2}t} \cI_d(\sigma^{-1}|n|\sqrt{t}) \d t\\
&=\int_0^\infty F(t)f(t,n,\sigma) dt,
\end{align*}
where, in the third step, we have used the rotation invariance of the standard spherical measure.

This allows us to conclude that $f(t,n,\sigma)$ is the probability density function of $Y_n(\sigma)$.
\end{proof}

%

We are now ready to complete the proof of Theorem \ref{thm:SINR}.

\begin{proof}[Proof of Theorem \ref{thm:SINR}]
(i)
In what follows, for simplicity, we will suppress from the notation the dependence on $\sigma$.
Since $Y_n(\sigma) = |X_n(\sigma)/\sigma|^2$, it is clear that
\[
p_c(\theta, \beta, \sigma) = \E\bigg[\prod_{j \in \La_B} \Big(1 + \theta\frac{|X_B|^{d\beta}}{|X_j|^{d\beta}} \Big)^{-1}\bigg]
=\E\bigg[\prod_{j \in \La_B} \Big(1 + \theta \frac{Y_B^{d\beta/2}}{Y_j^{d\beta/2}} \Big)^{-1}\bigg]. 
\]
Then it is straightforward to express the coverage probability as
 follows:
\begin{align}
p_c(\theta, \beta, \sigma)&= \sum_{n \in \L} \E\bigg[\prod_{j \in \La_n} \Big(1 + \theta
 \frac{Y_n^{d\beta/2}}{Y_j^{d\beta/2}} \Big)^{-1}; B=n\bigg]\nonumber\\
&=\sum_{n \in \L} \int_0^\infty  \E\bigg[\prod_{j \in \La_n} \Big(1 +
 \theta \frac{t^{d\beta/2}}{Y_j^{d\beta/2}} \Big)^{-1}; Y_j
 \ge t \ (\forall j \not=n) \bigg]
 f(t,n,\sigma) \d t \nonumber\\
&= \sum_{n \in \L} \int_0^\infty \prod_{j \in \La_n} \E\bigg[ \Big(1 +
 \theta \frac{t^{d\beta/2}}{Y_j^{d\beta/2}} \Big)^{-1}; Y_j
 \ge t \bigg]
 f(t,n,\sigma) \d t \label{coverage-for-Y}.
\end{align}
We complete the proof of (i) by noting that
\[
\E\bigg[ \Big(1 + \theta \frac{t^{d\beta/2}}{Y_j^{d\beta/2}} \Big)^{-1}; Y_j \ge t\bigg]
=\int_{t }^{\infty}  \Big(1 + \th \frac{t^{d\beta/2}}{u^{d\beta/2}} \Big)^{-1}f(u, j, \sigma) \d u.
\]

(ii) Making a change of variables $t=\theta^{-2/d\beta} s$ in the integral~\eqref{coverage-for-Y}, we have
\begin{align*}
p_c(\theta, \beta, \sigma) &= \sum_{n \in \L}  \int_0^\infty \prod_{j \in \La_n} \E\bigg[ \Big(1 + \frac{s^{d\beta/2}}{Y_j^{d\beta/2}} \Big)^{-1}; Y_j \ge \theta^{-2/d\beta}s\bigg] f(\theta^{-2/d\beta}s ,n,\sigma) \theta^{-2/d\beta} \d s.
\end{align*}
From Lemma \ref{lem:distribution-of-Y} and \eqref{eq:generalI}, we deduce that 
\begin{align}
\theta^{1/\beta} p_c(\theta, \beta, \sigma) 
&= \sum_{n \in \L} e^{-\frac{|n|^2}{2\sigma^2}} 
\int_0^\infty \prod_{j \in \La_n}
 \E\bigg[ \Big(1 +
 \frac{s^{d\beta/2}}{Y_j^{d\beta/2}} \Big)^{-1}; Y_j \ge
 \theta^{-2/d\beta}s\bigg] \nonumber \\
&\quad \times s^{\frac{d}{2}-1}
 \exp\left(-\frac{1}{2}\theta^{-2/d\beta}s\right)
 \cI_d(\sigma^{-1}|n|\sqrt{\theta^{-2/d\beta}s})
\d s \nonumber \\
&\to \sum_{n \in \L} e^{-\frac{|n|^2}{2\sigma^2}} 
\int_0^\infty \prod_{j \in \La_n}
 \E\bigg[ \Big(1 +
 \frac{s^{d\beta/2}}{Y_j^{d\beta/2}} \Big)^{-1}\bigg]
 \frac{s^{\frac{d}{2}-1}}{2^{d/2}\Gamma(d/2)} \d s 
\label{eq:larget}
\end{align}
as $\theta \to \infty$. We will justify the interchange of
 integration and limit in Subsection~\ref{sec:justify}. 

Recalling that the probability density of $Y_j$ is given by $f(\cdot,j,\sigma)$, we obtain the desired result.
\end{proof}

\subsubsection{Justification of the interchange of integration and
limit}\label{sec:justify}
We want to justify the interchange of integration and
limit in \eqref{eq:larget}. 
Let 
\[
f_{\delta}(n,s) 
:= \prod_{j \in \La_n}
 \E\bigg[ \Big(1 +
 \frac{s^{d\beta/2}}{Y_j^{d\beta/2}} \Big)^{-1}; Y_j \ge
 \delta s\bigg] s^{\frac{d}{2}-1}
 \exp\left(-\frac{1}{2}\delta s\right)
 \cI_d(\sigma^{-1}|n|\sqrt{\delta s})
\]
and 
\[
\int f(n,s) \mu(dn ds) := \sum_{n \in \La} e^{-\frac{|n|^2}{2\sigma^2}}
\int_0^{\infty} f(n,s) ds. 
\]
It is easy to see that 
\[
 f_{\delta}(n,s) \to \prod_{j \in \La_n} \E\bigg[ \Big(1 +
 \frac{s^{d\beta/2}}{Y_j^{d\beta/2}} \Big)^{-1}\bigg]
 \frac{s^{\frac{d}{2}-1}}{2^{d/2}\Gamma(d/2)}. 
\]
By the dominated convergence theorem, it suffices to show that for some $\delta_0>0$ 
\[
\int \sup_{0 < \delta < \delta_0} |f_{\delta}(n,s)| \mu(dn
ds) < \infty. 
\]
\begin{lem} Let $d \ge 2$. There exist $C_1, C_2 > 0$ such
 that 
\begin{align*}
\prod_{j \in \La_n}
 \E\bigg[ \Big(1 + \frac{s^{d\beta/2}}{Y_j^{d\beta/2}} \Big)^{-1} \bigg] 
\le C_1 e^{-C_2 s} \quad \text{for $s \ge 0$}. 
\end{align*}
Here $C_1$ and $C_2$ may depend on $\sigma, d, \beta$ but not $n$. 
\end{lem}
\begin{proof}
We recall that $Y_j = |j/\sigma + \eta|^2$ where $\eta$ is a $d$-dimensional
 standard Gaussian random variables. Then, $\E Y_j = |j/\sigma|^2 + d$. 
For $0 < c_1 < c_2$, let 
\[
 \La_{n,c_1,c_2} := \{{j \in \La_n : c_1 s \le \E Y_j \le c_2 s}\}
\]
and 
for $a>0$, which is chosen later, let 
\[
 B_a
= \bigcap_{j \in \La_{n,c_1,c_2}} \{Y_j \le a \E Y_j\}.  
\]
On $B_a$, since $Y_j \le a\E Y_j \le ac_2 s$ for $j \in \La_{n,c_1,c_2}$, we have 
\begin{align}
\E\bigg[\prod_{j \in \La_{n,c_1,c_2}} 
\Big(1 + \frac{s^{d\beta/2}}{Y_j^{d\beta/2}} \Big)^{-1} ; B_a\bigg] 
&\le \E\bigg[\prod_{j \in \La_{n,c_1,c_2}} 
\Big(1 + (ac_2)^{-d\beta/2} \Big)^{-1} ; B_a\bigg] \nonumber\\  
&\le e^{-\alpha c_3 s^{d/2}}. 
\label{eq:ba}
\end{align}
where $\alpha = \log (1+(ac_2)^{-d\beta/2}) > 0$ and 
we used the fact that $|\La_{n,c_1,c_2}| \ge c_3 s^{d/2}$ 
for any sufficiently large $s$. On $B_a^{\complement}$, for any sufficiently large $s$, 
\begin{align*}
 \E\bigg[\prod_{j \in \La_{n,c_1,c_2}} 
\Big(1 + \frac{s^{d\beta/2}}{Y_j^{d\beta/2}} \Big)^{-1} ; B_a^{\complement} \bigg] 
&\le \P( B_a^{\complement} ) 
\le \sum_{j \in \La_{n,c_1,c_2}} P(Y_j > a c_1s). 
\end{align*}
Since $|j/\sigma|^2 \le c_2 s$ for $j\in \La_{n,c_1,c_2}$, we see that 
\begin{align*}
a c_1s < Y_j = |j/\sigma + \eta|^2 \le 2(|j/\sigma|^2 + |\eta|^2) 
\le 2(c_2s + |\eta|^2) 
\end{align*}
and, by taking $a = 2(c_2 +1)/c_1$, we have 
\begin{align*}
P(Y_j> a c_1s) 
&\le P\left( 
\frac{1}{2} (a c_1 - 2 c_2) s \le
 |\eta|^2 
\right) \\
&= P(|\eta|^2 > s) \\
& \sim \frac{1}{\Gamma(d/2)} (\frac{s}{2})^{d/2-1}e^{-s/2}
 \quad \text{as $s \to \infty$}, 
\end{align*}
where $|\eta|^2$ is the $\chi^2_d$-distribution with $d$
 degrees of freedom (see \eqref{eq:chi-squared}). 
Thus, 
\begin{align}
\E\bigg[\prod_{j \in \La_{n,c_1,c_2}} 
\Big(1 + \frac{s^{d\beta/2}}{Y_j^{d\beta/2}} \Big)^{-1} ; B_a^{\complement} \bigg] 
&\le c_5 s^{d-1}e^{-s/2}.  
\label{eq:bac}
\end{align}
Therefore, \eqref{eq:ba} with $d \ge 2$ and \eqref{eq:bac} yield 
\begin{align*}
\prod_{j \in \La_n} 
\E\bigg[ \Big(1 + \frac{s^{d\beta/2}}{Y_j^{d\beta/2}} \Big)^{-1} \bigg] 
&\le \E\bigg[\prod_{j \in \La_{n,c_1,c_2}}
 \Big(1 + \frac{s^{d\beta/2}}{Y_j^{d\beta/2}} \Big)^{-1} \bigg] \\
&\le C_1\exp(-C_2 s).  
\end{align*}
This completes the proof. 
\end{proof}
Since $\cI_d(u) \le
\cI_d(0) e^u$ 
for every $u \ge 0$, we have 
\begin{align}
|f_{\delta}(n,s)|
&= 
\prod_{j \in \La_n}
 \E\bigg[ \Big(1 + \frac{s^{d\beta/2}}{Y_j^{d\beta/2}} \Big)^{-1}\bigg] s^{\frac{d}{2}-1}
 e^{-\delta s/2}\cI_d(\sigma^{-1}|n|\sqrt{\delta s})
 \nonumber \\
&\le C_1 e^{-C_2s} \cdot s^{\frac{d}{2}-1} \cI_d(0) e^{b\sqrt{\delta s}} 
\le C e^{-cs} e^{b\sqrt{\delta s}}, 
\label{eq:integraloffdelta}
\end{align}
where $b = |n|/\sigma$. 
From \eqref{eq:integraloffdelta}, by taking $\delta_0$ small enough, we have 
\begin{align*}
\int_0^{\infty} \sup_{0 < \delta < \delta_0} 
|f_{\delta}(n,s)| ds
&\le C \int_0^{\infty} e^{-c s} e^{b \sqrt{\delta_0 s}} ds\\
&= C\left\{
\frac{1}{c} + \frac{b \sqrt{\delta_0} e^{\frac{\delta_0}{4c}
b^2} \sqrt{\pi}}{2c^{3/2}}
\left(1+\mathrm{Erf}\Big( \frac{b\sqrt{\delta_0}}{2\sqrt{c}}
 \Big) \right)
\right\}\\
& \le C' e^{\frac{b^2}{4}} 
= C' e^{\frac{|n|^2}{4\sigma^2}}, 
\end{align*}
where $\mathrm{Erf}(z) = \frac{2}{\sqrt{\pi}}\int_0^z
e^{-t^2}dt$. 
Hence 
\[
 \int \sup_{0< \delta < \delta_0} |f_{\delta}(n,s)|\mu(dn ds) 
\le C' \sum_{n \in \La} e^{-\frac{|n|^2}{4\sigma^2}} < \infty, 
\]
which justifies that the order of limit and integration can be interchanged. 

\subsection{Proof of Theorem~\ref{result:conv}} \label{sec:thmproof71}
\begin{proof}[Proof of Theorem~ \ref{result:conv}]
We use the following expression, which holds for any point process $\La$: \[p_c(\t,\La)= \E_\La \left[ \prod_{\la \in \La_\B} \left(1 + \theta \cdot  \frac{|X_\B|^{d\b}}{|X_\la|^{d\b}} \right)^{-1} \right]  \]
Let $B(0;R)$ denote the ball of radius $R$ in $\R^d$. Fix such an $R>0$, to be thought of as large. Consider a compactly supported smooth radial function $\phi:\R^d \to [0,1]$, such that $\phi(x)=1$ if $x \in B(0;R)$, 
and $\phi(x)$ decreases to 0 as $|x| \to \infty$. Consider
 the following functional $p_c(\t,\La,R)$ that is associated
 with $p_c(\t,\La)$ \[p_c(\t,\La,R)= \E_\La \left[ \prod_{\la \in \La_\B} \left(1 + \theta \cdot \frac{|X_\B \phi(X_\la)|^{d\b}}{|X_\la|^{d\b}} \right)^{-1} \right] .\]

Observe that the functional $\prod_{\la \in \La_\B} \left(1 + \theta \cdot \frac{|X_\B \phi(X_\la)|^{d\b}}{|X_\la|^{d\b}} \right)^{-1}$ is continuous on the space of point configurations on $\R^d$, where the latter is endowed with the \textit{vague topology} on the space of locally finite point configurations on $\R^d$, entailing convergence as discrete measures on compact sets. Furthermore, this functional is bounded by 1. 
These two facts, together with convergence in distribution of $\Xi_\s$ to $\Xi$ 
implies that, 
for any fixed $R>0$, we have $p_c(\t,\Xi_\s,R) \to p_c(\t,\Xi,R)$ as $\s \to \infty$. Now, let $\La \in \{\Xi, \{\Xi_\s\}_{\s>0} \}$. Let $\La_{B}^R$ denote $\La \cap B(0;R)$ and $\La_B^{R^\complement}$ denote 
$\La \cap B(0;R)^\complement$.
 
Let $g_{\t}(x) = (1+\theta |x|^{d\beta})^{-1}$. Then we have,
\begin{align*}
\lefteqn{|p_c(\t,\La)-p_c(\t,\La,R)|} \\
&= \left| \E_\La \left[ \prod_{\la \in \La_\B^R} g_{\t}\Big(\frac{|X_\B|}{|X_\la|}\Big) \left( \prod_{\la \in \La_\B^{R^\complement}} 
g_{\t}\Big(\frac{|X_\B|}{|X_\la|}\Big)   
-\prod_{\la \in \La_\B^{R^\complement}} g_{\t}\Big(\frac{|X_\B \phi(X_{\la})|}{|X_\la|}\Big) \right) 
 \right] \right| \\
&\le \E_\La \left[ \Big| \prod_{\la \in \La_\B^{R^\complement}} g_{\t}\Big(\frac{|X_\B|}{|X_\la|}\Big) 
 - \prod_{\la \in \La_\B^{R^\complement}} 
g_{\t}\Big(\frac{|X_\B \phi(X_{\la})|}{|X_\la|}\Big) \Big| \right] \\
&\le \E_\La \left[ \Big| \prod_{\la \in \La_\B^{R^\complement}} 
g_{\t}\Big(\frac{|X_\B|}{|X_\la|}\Big) -1 \Big| \right] + 
\E_\La \left[ \Big| \prod_{\la \in \La_\B^{R^\complement}} 
g_{\t}\Big(\frac{|X_\B \phi(X_{\la})|}{|X_\la|}\Big) -1 \Big| \right] \\
\numberthis \label{eq:conv-diff}
 \end{align*} 
Now we observe that 
\[ 
\prod_{\la \in \La_\B^{R^\complement}} 
g_{\t}\Big(\frac{|X_\B|}{|X_\la|}\Big)^{-1} 
\le 
 \prod_{\la \in \La_\B^{R^\complement}} \exp\left(\t \frac{|X_\B|^{d\b}}{|X_\la|^{d\b}}  \right) = \exp \left( \t |X_\B|^{d\b}\Sigma_R^\b(\La) \right). 
\] 
Since $|\phi(X_\la)|\le 1$, we may also deduce that
 \[  
\prod_{\la \in \La_\B^{R^\complement}} 
g_{\t}\Big(\frac{|X_\B \phi(X_{\la})|}{|X_\la|}\Big)^{-1} 
\le \exp \left( \t |X_\B|^{d\b}\Sigma_R^\b(\La) \right).      
\]
This implies that 
\begin{equation} \label{eq:conv-min} 
\min\left\{ \prod_{\la \in \La_\B^{R^\complement}} 
g_{\t}\Big(\frac{|X_\B|}{|X_\la|}\Big),  
 \prod_{\la \in \La_\B^{R^\complement}} 
g_{\t}\Big(\frac{|X_\B \phi(X_{\la})|}{|X_\la|}\Big) 
\right\} 
\ge \exp \left( -\t |X_\B|^{d\b}\Sigma_R^\b(\La) \right). 
 \end{equation}
 On the other hand,
 \begin{equation} \label{eq:conv-max} 
\max\left\{ \prod_{\la \in \La_\B^{R^\complement}} 
g_{\t}\Big(\frac{|X_\B|}{|X_\la|}\Big),  
 \prod_{\la \in \La_\B^{R^\complement}} 
g_{\t}\Big(\frac{|X_\B \phi(X_{\la})|}{|X_\la|}\Big) 
\right\} 
\le 1. 
 \end{equation}
As a consequence of \eqref{eq:conv-min} and \eqref{eq:conv-max}, 
we deduce that 
\begin{align*}
&\E_\La \Big[ \Big| 
\prod_{\la \in \La_\B^{R^\complement}} g_{\t}\Big(\frac{|X_\B|}{|X_\la|}\Big) -1 
\Big| \Big] 
+ \E_\La \Big[ \Big| \prod_{\la \in \La_\B^{R^\complement}} 
g_{\t}\Big(\frac{|X_\B \phi(X_{\la})|}{|X_\la|}\Big) -1 \Big| \Big] \\
&\quad \le 2\E_\La \Big[ 1 - \exp \big( -\t |X_\B|^{d\b}\Sigma_R^\b(\La) \big) 
\Big],  
\end{align*}
which converges to 0 as $R \to \infty$ uniformly over $\La \in \{\Xi, \{\Xi_\s\}_{\s>0}\}$ because of the uniform convergence  $\Sigma_R^\b(\La) \to 0$ (as $R \to \infty$). 

In light of \eqref{eq:conv-diff}, this implies that we can choose $R$ and $\s$ (depending on $R$) large enough such that
 \begin{align*}
\lefteqn{|p_c(\t,\Xi_\s)-p_c(\t,\Xi)|} \\ 
&\le |p_c(\t,\Xi_\s)-p_c(\t,\Xi_\s,R)| 
+ |p_c(\t,\Xi_\s,R)-p_c(\t,\Xi,R)| 
+ |p_c(\t,\Xi,R)-p_c(\t,\Xi)|  
\end{align*} 
 can be made arbitrarily small. Thus, as $\s \to \infty$, we have the convergence of coverage probabilities $p_c(\t,\Xi_\s) \to p_c(\t,\Xi)$.  This completes the proof of our result.
\end{proof}

\section{Acknowledgements}
We would like to thank Joel L. Lebowitz and Khanh Duy Trinh for illuminating discussions regarding this work, and the anonymous referees for their insightful comments and suggestions. SG was supported in part by the MOE grants R-146-000-250-133, R-146-000-312-114 and MOE-T2EP20121-0013. 
NM was supported in part by the Grant-in-Aid for Scientific Research (C) (No.19K11838) 
of the Japan Society for the Promotion of Science (JSPS).
TS was supported in part by the Grant-in-Aid for Scientific Research (B) (No.18H01124) and (S) (No.16H06338) of Japan Society for the Promotion of Science (JSPS) and by 
JST CREST Mathematics (15656429). 

\appendix 

\begin{center}
{\Large   \textbf{Appendix}}
\end{center} 
 
\section{Appendix : Persistent homology and persistence diagrams} \label{a:PD}
In this appendix, we briefly recall the definition of
persistent homology and persistence diagram. We refer the
readers to
(cf. \cite{zomorodian2005computing,edelsbrunner2008persistent,weinberger2011persistent,
HST18}) and references therein for more details. 

Let $V$ be a finite set and we simply write $V=\{1,2,\dots,n\}$. 
An abstract simplicial complex over $V$ is a 
collection $K$ of subsets which is closed under the operation of taking
non-empty subsets, i.e., if $\sigma \in K$ and
$(\emptyset \not=)\eta \subset \sigma$, then $\eta \in K$. 
An element $\sigma \in K$ is called a $q$-simplex if the
cardinality of $|\sigma|=q+1$. A $0$-simplex is a vertex, 
$1$-simplex is an edge, 
$2$-simplex is a face, and so on. 
We denote the set of $q$-simplices by $K_q$. 
We introduce an equivalence relation on ordered simplices by 
$(i_0, \dots, i_q) \sim (j_0, \dots, j_q)$ if there exists
an even permutation $\pi \in \mathcal{S}_{q+1}$ such that $j_k =
\pi(i_k)$ for every $k=0,1,\dots,q$. There are two
equivalence classes, which are called orientations. 
When $\sigma = \{i_0, \dots, i_q\}$, we denote by $\bra
\sigma \ket$ the equivalence class to which the ordered
simplex $(i_0, \dots, i_q)$ belongs. 
When $(i_0, \dots, i_q)$ and $(j_0, \dots, j_q)$ have
different orientations, we put minus sign as 
$\bra i_0, \dots, i_q \ket = - \bra j_0, \dots, j_q \ket$. 
For example, $\bra 1, 3, 2\ket = -\bra 1,2,3 \ket$, 
$\bra 1, 4, 3, 2 \ket = - \bra 1,2,3,4\ket$. 
Let $\F$ be a field and we define the $q$-chain group by
\[
 C_q(K) := \{\sum_{\sigma \in K_q} a_{\sigma} \bra \sigma
 \ket : a_{\sigma} \in \F\}, 
\]
which is an $\F$-vector space having a basis $\{\bra \sigma
\ket : \sigma \in K_q\}$.  
Let $\partial_q : C_q(K) \to C_{q-1}(K)$ be a linear map
(boundary map) defined by
\[
 \partial_q \bra i_0,i_1,\dots,i_q \ket
= \sum_{k=0}^q (-1)^q \bra i_0,\dots,\widehat{i_k} \dots, i_q \ket 
\]
where $\widehat{i_k}$ means the removal of $i_k$, and it is
extended linearly for general elements of $C_q(K)$ by 
$\partial_q(\sum_{\sigma \in K_q} a_{\sigma} \bra \sigma
 \ket)
= \sum_{\sigma \in K_q} a_{\sigma} \partial_q(\bra \sigma
 \ket)$. 
It is easy to see that $\partial_{q} \circ \partial_{q+1} = 0$ for
every $q$. 
The chain groups and boundary maps are assembled into a
chain complex: 
\[
\cdots \to  C_q(K) \stackrel{\partial_q}{\to} C_{q-1}(K)
\stackrel{\partial_{q-1}}{\to} 
\cdots \stackrel{\partial_2}{\to} C_1(K)
\stackrel{\partial_1}{\to} C_0(K) 
\stackrel{\partial_0}{\to} 0 
\]
From the relation $\partial_{q} \circ \partial_{q+1}$, we have 
$\im \partial_{q+1} \subset \ker \partial_{q}$. Then one can
define the $q$th-homology group by 
\[
 H_q(K) := \ker \partial_q / \im \partial_{q+1}, 
\]
which is the quotient $\F$-vector space. Each generator of
$H_q(K)$ represents a $q$-dimensional hole. 

A filtration $\K = \{K_t\}_{t=0}^n$ of simplicial complexes is an incresing family 
of simplicial complexes $K_0 \subset K_1 \subset \cdots
\subset K_n$. For each simplical complex $K_t$, the homology
group $H_q(K_t)$ is defined as above. So a family of
homology groups $\{H_q(K_t)\}_{t=0}^n$ has an information
about how $q$-dimensional holes change as $t$ varies. 
But, one can observe more information, persistence, by
employing persistent homology. Since we have a filtration,
we have an inclusion $\iota_t : K_t \to K_{t+1}$, which
induces the linear maps $(\iota_t)_*$ on homology groups
$\{H_q(K_t)\}_{t=0}^n$: 
\begin{equation}
H_q(K_0) \stackrel{(\iota_0)_*}{\to} \cdots 
\stackrel{(\iota_{t-1})_*}{\to} 
H_q(K_t) \stackrel{(\iota_t)*}{\to} H_q(K_{t+1})
\stackrel{(\iota_{t+1})_*}{\to} \cdots 
\stackrel{(\iota_{T-1})*}{\to} H_q(K_n), 
\label{eq:persistent_module} 
\end{equation}
which is called a \textit{persistence module}, denoted by $H_q(\K)$.  
A collection of vector spaces
$\{V_t\}_{t=0}^n$ and linear maps $f_{st} : V_s \to V_t$
such that $f_{st} \circ f_{tu} = f_{su}$ for $s<t<u$ 
is called a representation of $A_n$-quiver. 
In the persistence module above, $f_{st} = (\iota_t)\circ
\cdots \circ (\iota_s)_*$ and $V_t = H_q(K_t)$. 
A basic representation of $A_n$-quiver is an \textit{interval
module} $I(b,d)$ defined by 
\[
0 \stackrel{g_0}{\to} \cdots \to 0 \stackrel{g_{b-1}}{\to} \F \stackrel{g_b}{\to} \F \stackrel{g_{b+1}}{\to} \cdots \to \F \stackrel{g_{d-1}}{\to} 0 \to
\cdots \stackrel{g_n}{\to} 0,  
\]
where $g_t = \mathrm{id}_{\F}$ is the identity map on $\F$ for $t=b,b+1,\dots, d-2$. 
From the theory of representations of $A_n$-quiver, we have the following decomposition property. 
\begin{thm}(\cite{zomorodian2005computing}) \label{thm:PH}
Let $H_q(\K)$ be a persistence module of the form
 \eqref{eq:persistent_module}.  
There uniquely exist indices $p \in \Z_{\ge 0}$ and
 $b_i, d_i \in \{0,1,\dots, n\}$ with $0 \le b_i < d_i \le n, i = 1, 2, \dots,
 p$ such that the following isomorphism holds:
\begin{equation}
 H_q(\K) \simeq \bigoplus_{i=1}^p I(b_i,d_i)
\label{eq:PH}
\end{equation}
\end{thm}
We can interpret the interval module $I(b_i,d_i)$ 
as persistence of the $i$th generator of $q$th homology group, which appears at
time $b_i$ and disappears at time $d_i$. 
In order to visualize the information \eqref{eq:PH} of the decomposition
of $H_q(\K)$, 
we denote its persistence diagram by 
\[
D_q(\K) := \{(b_i,d_i) : i=1,2\dots, p\} \subset \Delta, 
\]
which is a multiset in $\Delta:=\{(x,y) \in \R^2 :0 \le  x <
y \le \infty\}$. See Fig.~\ref{fig:PD2}. 

\begin{ex}
We consider an example of a filtration given in
 Fig.~\ref{fig:filtration2}. 
We can observe that 
there is a cycle ($1$-dimensional hole) at time $1$, 
two independent cycles at time $2$, and 
a cycle at time $3$. It is also considered that 
a cycle appears at time $1$, another cycle appears at time
 $t=2$, one cycle disappears at time $3$ and $4$. This
 reflects on the persistent module (over a field $\R$) as 
\begin{align*}
H_1(\K) &:=  H_1(K_0) \to H_1(K_1) \to  H_1(K_2) \to H_1(K_3) \to
 H_1(K_4) \\
& \simeq 0 \to \R \to \R^2 \to \R \to 0. 
\end{align*}
In this case, we can see that 
\[
 H_1(\K) \simeq I(1,4) \oplus I(2,3) 
\simeq (0 \to \R \to \R \to \R \to 0) 
\oplus (0 \to 0 \to \R \to 0 \to 0), 
\]
although we have another possibility of decomposition by
 interval modules as 
\[
 H_1(\K) \simeq I(1,3) \oplus I(2,4) 
\simeq (0 \to \R \to \R \to 0\to 0) 
\oplus (0 \to 0 \to \R \to \R \to 0). 
\]
One cannot observe this distinction of decompositions if 
one only has the homology groups $\{H_q(K_t)\}_{t=0}^n$. This is
 a consequence of induced linear maps
 $\{(i_t)_*\}_{t=0}^n$ which comes from the inclusions of
 the filtration $\K$. So the persistent homology has more
 information than the collection of homology groups do. 
\begin{figure}[htbp]
\begin{center}
\includegraphics[width=0.8\hsize]{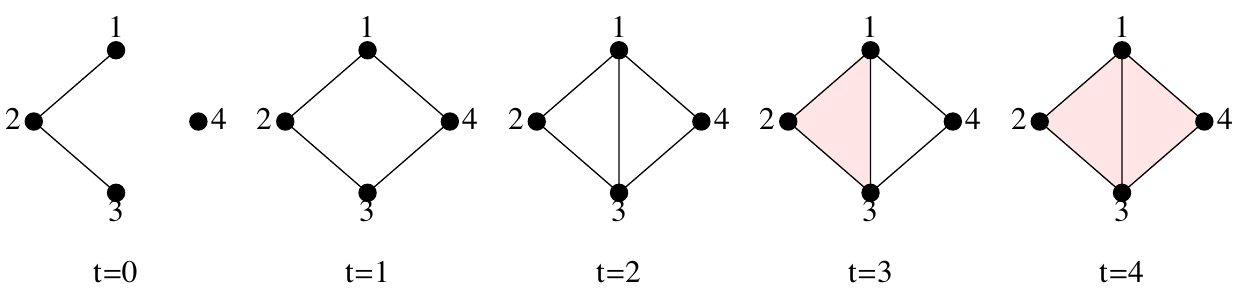}
\end{center}
\caption{Filtration. $\{1,2,3\}$ and $\{1,3,4\}$ are
 $2$-simplices and colored in pink at time $3$ and $4$.}
\label{fig:filtration2}
\end{figure}
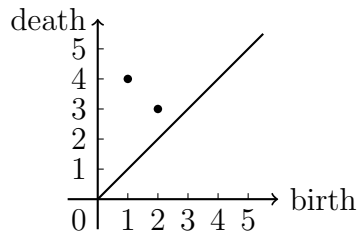
\begin{figure}
\begin{center}
\begin{tikzpicture}[scale =0.4]
\draw[thick, ->] (-1,0) -- (6,0) node [right] {$\mathrm{birth}$};
\draw[thick, ->] (0,-1) -- (0,6) node [left] {$\mathrm{death}$}; 
\draw[thick, -] (0,0) -- (5.5,5.5);  
\fill (1,4) circle (4pt); 
\fill (2,3) circle (4pt); 
\node at (0,0) [anchor=north east] {0};
\foreach \x in {1,2,3,4,5}
   \draw (\x cm,.2) -- (\x cm,0) node[anchor=north] {$\x$};
\foreach \y in {1,2,3,4,5}
   \draw (.2,\y cm) -- (0,\y cm) node[anchor=east] {$\y$};
\end{tikzpicture}
\end{center} 
\caption{Persistence diagram of $H_1(\K)$ for the filtration given in Fig.~\ref{fig:filtration2}.}
\label{fig:PD2}
\end{figure}
We remark that here in this example, we treat index $t$ as time, but it is
 sometimes treated as resolution or other parameters,
 depending on what we are looking at. In the \v{Cech}
 filtration defined below, the radius $r$ is such a
 parameter. 
\end{ex}

Given $\Phi = \{x_1,x_2, \dots,x_N\} \subset \R^d$ and $r
>0$, we define a \v{C}ech complex built over $\Phi$ by 
\[
 K(\Phi, r) := \{\sigma \subset \Phi : \bigcap_{x \in \sigma}
 B(x,r) \not= \emptyset\}, 
\]
where $B(x,r)$ is the ball of radius $r$ centered at $x$. 
A collection $\K(\Phi) := \{K(\Phi,r)\}_{r \ge 0}$ is an increasing
family of \v{C}ech complexes, which we call \v{C}ech
filtration over $\Phi$. 
From the \v{C}ech filtration, we have the persistent homology
$H_q(\K(\Phi))$ by Theorem~\ref{thm:PH} and its persistence
diagram. 
Here we deal with a filtration with continuous parameter $r$, but
for a finite configuration $\Phi$, the birth of simplices
occurs only at discrete finite parameters $r_1 < r_2 <
\cdots < r_n$ so that 
we can apply Theorem~\ref{thm:PH} to this case and define
its persistence diagrams supported on $\{r_1,\dots, r_n\}
\cap \Delta$ in an obvious way. 

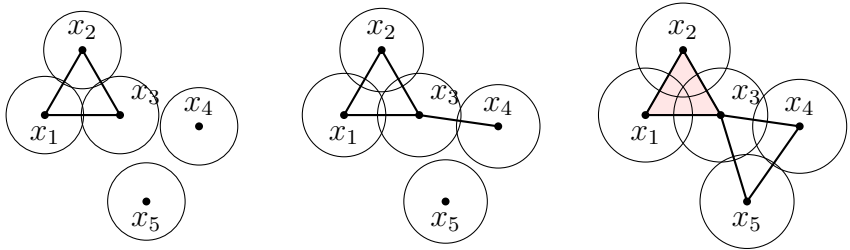
\begin{figure}[!h]
\begin{center}
\begin{minipage}{3.8cm}
\begin{tikzpicture}[scale = 0.5]
\coordinate (1) at (0,0);
\coordinate (2) at (1, 1.73);
\coordinate (3) at (2,0);
\coordinate (4) at (4.1, -0.3);
\coordinate (5) at (2.7, -2.3);
\foreach \n in {1, ..., 5}
\fill (\n) circle (3pt); 
\foreach \n in {1, ..., 5}
\draw (\n) circle [radius = 1.03cm]; 
\draw[thick] (1) -- (2);  
\draw[thick] (1) -- (3);  
\draw[thick] (2) -- (3);  
\node [below] at (1) {$x_1$}; 
\node [above] at (2) {$x_2$}; 
\node [above right] at (3) {$x_3$}; 
\node [above] at (4) {$x_4$}; 
\node [below] at (5) {$x_5$}; 
\end{tikzpicture}
\end{minipage}
\begin{minipage}{3.8cm}
\begin{tikzpicture}[scale = 0.5]
\coordinate (1) at (0,0);
\coordinate (2) at (1, 1.73);
\coordinate (3) at (2,0);
\coordinate (4) at (4.1, -0.3);
\coordinate (5) at (2.7, -2.3);
\foreach \n in {1, ..., 5}
\fill (\n) circle (3pt); 
\foreach \n in {1, ..., 5}
\draw (\n) circle [radius = 1.11cm]; 
\draw[thick] (1) -- (2);  
\draw[thick] (1) -- (3);  
\draw[thick] (2) -- (3);  
\draw[thick] (3) -- (4);  
\node [below] at (1) {$x_1$}; 
\node [above] at (2) {$x_2$}; 
\node [above right] at (3) {$x_3$}; 
\node [above] at (4) {$x_4$}; 
\node [below] at (5) {$x_5$}; 
\end{tikzpicture}
\end{minipage}
\begin{minipage}{3.8cm}
\begin{tikzpicture}[scale = 0.5]
\coordinate (1) at (0,0);
\coordinate (2) at (1, 1.73);
\coordinate (3) at (2,0);
\coordinate (4) at (4.1, -0.3);
\coordinate (5) at (2.7, -2.3);
\filldraw [fill=pink!40] (1) -- (2) -- (3); 
\foreach \n in {1, ..., 5}
\fill (\n) circle (3pt); 
\foreach \n in {1, ..., 5}
\draw (\n) circle [radius = 1.25cm]; 
\draw[thick] (1) -- (2);  
\draw[thick] (1) -- (3);  
\draw[thick] (2) -- (3);  
\draw[thick] (3) -- (4);  
\draw[thick] (3) -- (5);  
\draw[thick] (4) -- (5);  
\node [below] at (1) {$x_1$}; 
\node [above] at (2) {$x_2$}; 
\node [above right] at (3) {$x_3$}; 
\node [above] at (4) {$x_4$}; 
\node [below] at (5) {$x_5$}; 
\end{tikzpicture}
\end{minipage}
\end{center} 
\caption{\v{C}ech complexes $K(\Phi, r)$ over
 $\Phi=\{x_1,x_2,\dots,x_5\}$ as $r$ increases. }
\end{figure}

\newpage

\section{Appendix : Comparison of nnd-s of 2D \& 3D point sets} \label{a:nnd}
\def\size3{0.5}
\begin{figure}[hbtp]
\begin{minipage}{0.52\hsize}
\begin{center}
\includegraphics[scale=\size3]{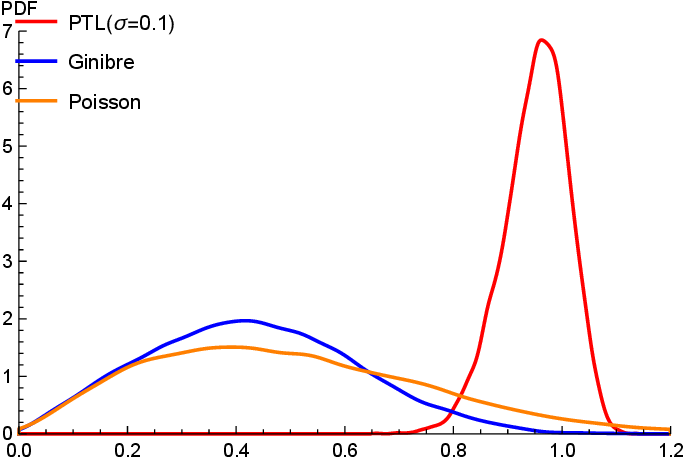}
\end{center}
\end{minipage}
\begin{minipage}{0.48\hsize}
\begin{center}
\includegraphics[scale=\size3]{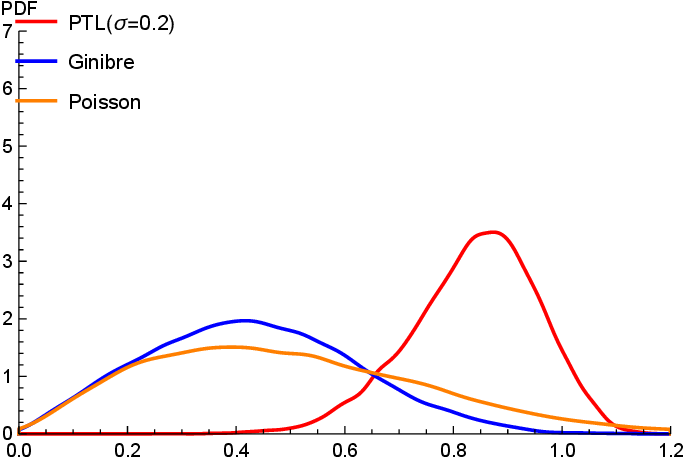}
\end{center}
\end{minipage}\\[4mm]
\begin{minipage}{0.52\hsize}
\begin{center}
\includegraphics[scale=\size3]{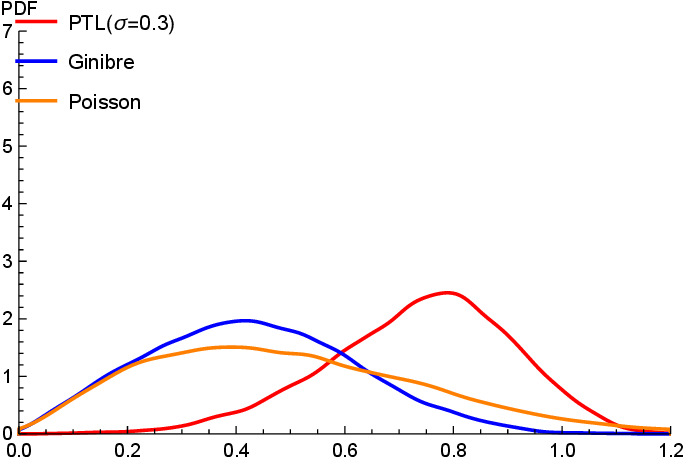}
\end{center}
\end{minipage}
\begin{minipage}{0.48\hsize}
\begin{center}
\includegraphics[scale=\size3]{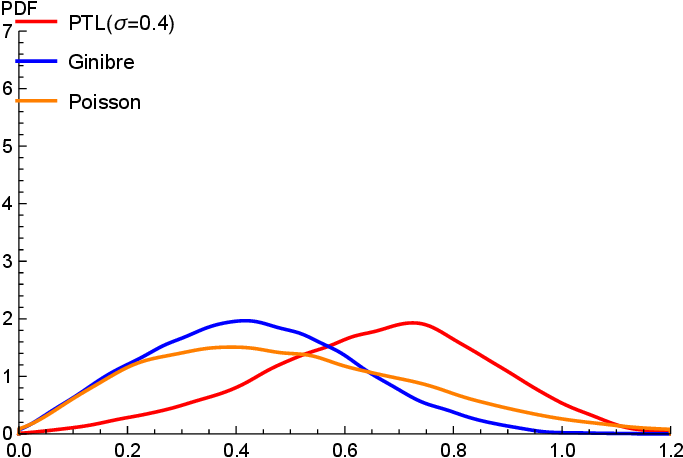}
\end{center}
\end{minipage}\\[4mm]
\begin{minipage}{0.52\hsize}
\begin{center}
\includegraphics[scale=\size3]{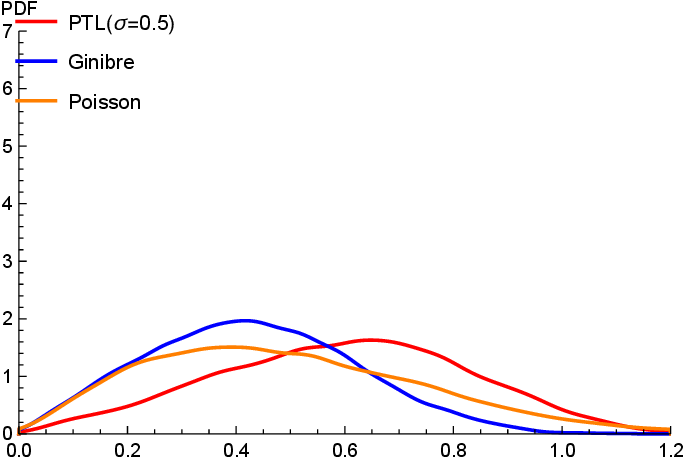}
\end{center}
\end{minipage}
\begin{minipage}{0.48\hsize}
\begin{center}
\includegraphics[scale=\size3]{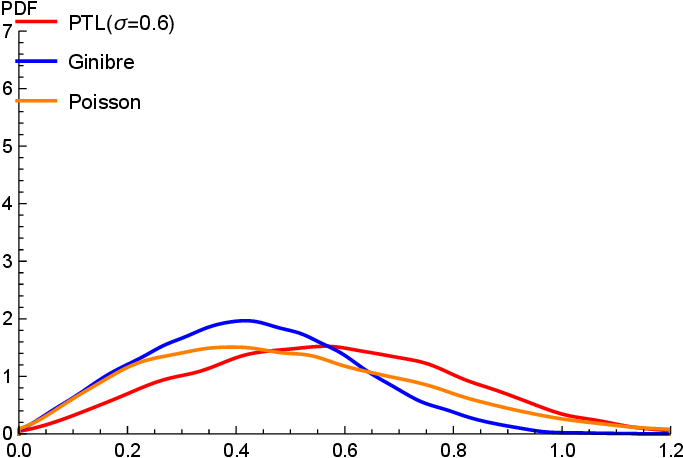}
\end{center}
\end{minipage}\\[4mm]
\begin{minipage}{0.52\hsize}
\begin{center}
\includegraphics[scale=\size3]{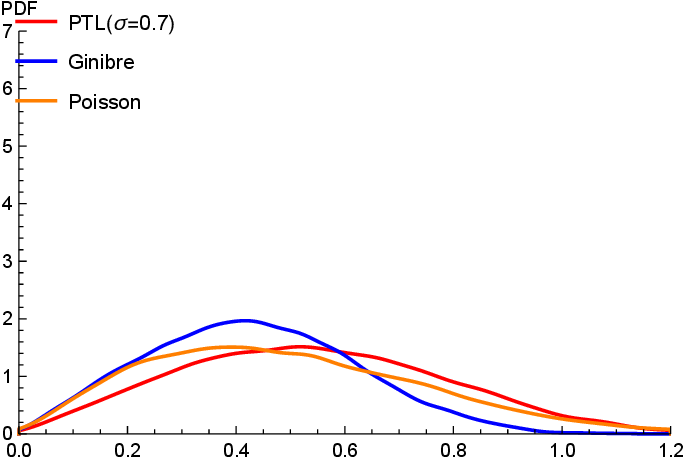}
\end{center}
\end{minipage}
\begin{minipage}{0.48\hsize}
\begin{center}
\includegraphics[scale=\size3]{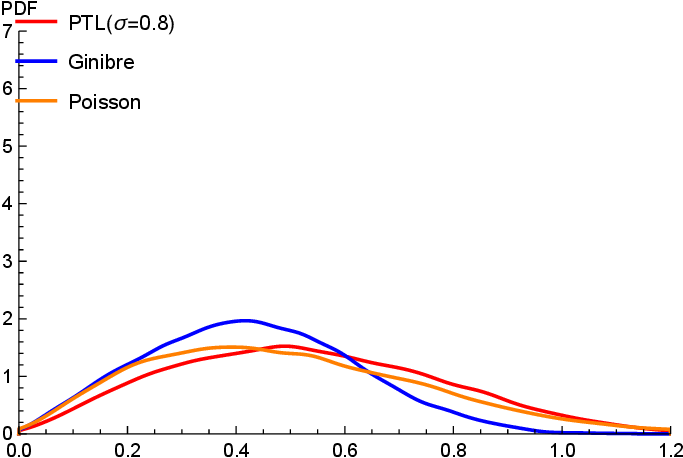}
\end{center}
\end{minipage}
\caption{Nearest neighbour distribution for Perturbed Triangular Lattice(PTL), Ginibre and
 Poisson ($\sigma=0.1, 0.2, \dots, 0.8$) for $\beta=2$. }
\label{fig:2dim-nearest-neighbour}
\end{figure}

\def\size4{0.55}
\begin{figure}[htbp]
\begin{minipage}{0.52\hsize}
\begin{center}
\includegraphics[scale=\size4]{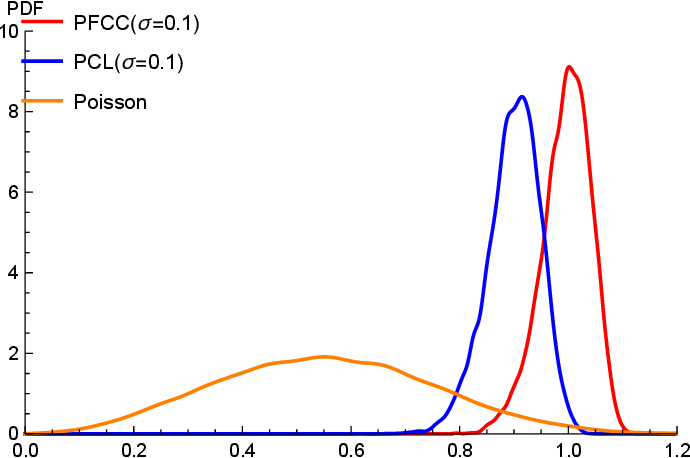}
\end{center}
\end{minipage}
\begin{minipage}{0.48\hsize}
\begin{center}
\includegraphics[scale=\size4]{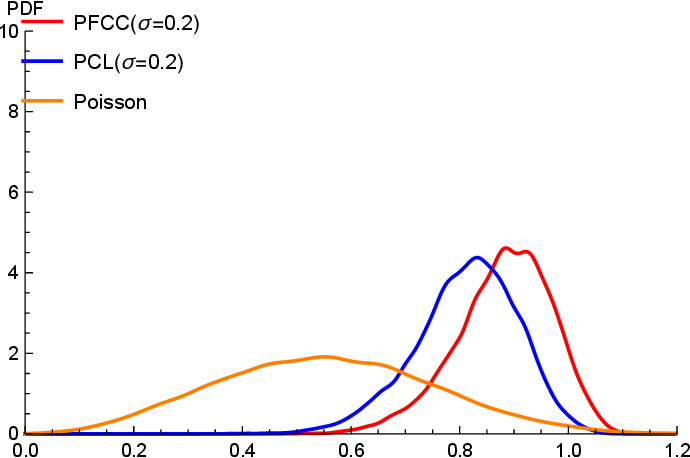}
\end{center}
\end{minipage}\\[4mm]
\begin{minipage}{0.52\hsize}
\begin{center}
\includegraphics[scale=\size4]{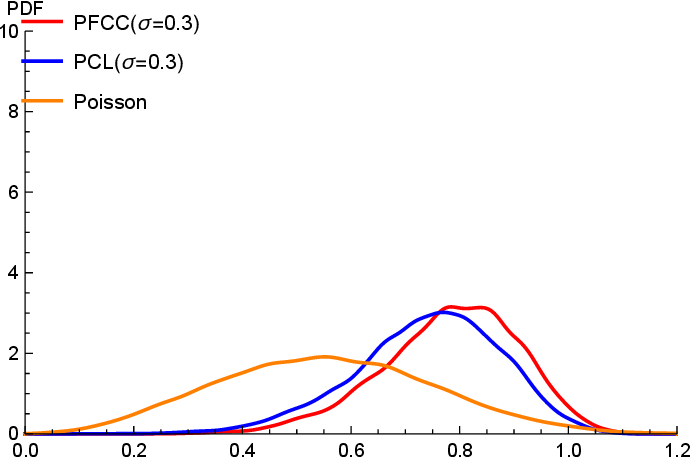}
\end{center}
\end{minipage}
\begin{minipage}{0.48\hsize}
\begin{center}
\includegraphics[scale=\size4]{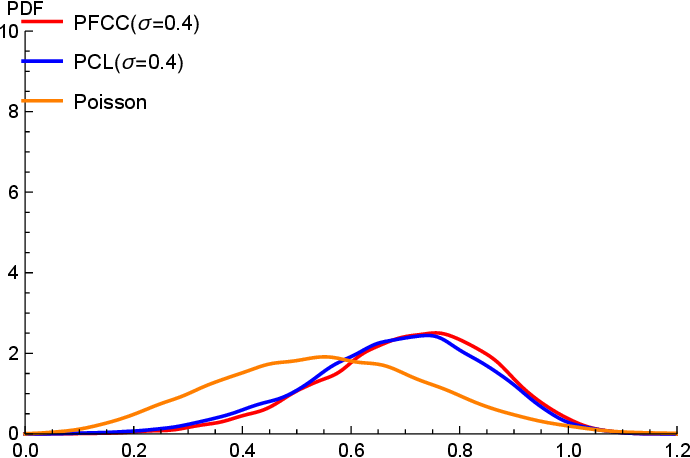}
\end{center}
\end{minipage}\\[4mm]
\begin{minipage}{0.52\hsize}
\begin{center}
\includegraphics[scale=\size4]{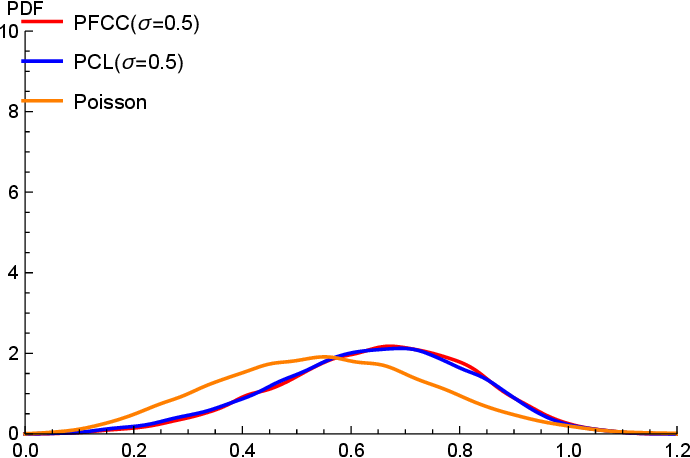}
\end{center}
\end{minipage}
\begin{minipage}{0.48\hsize}
\begin{center}
\includegraphics[scale=\size4]{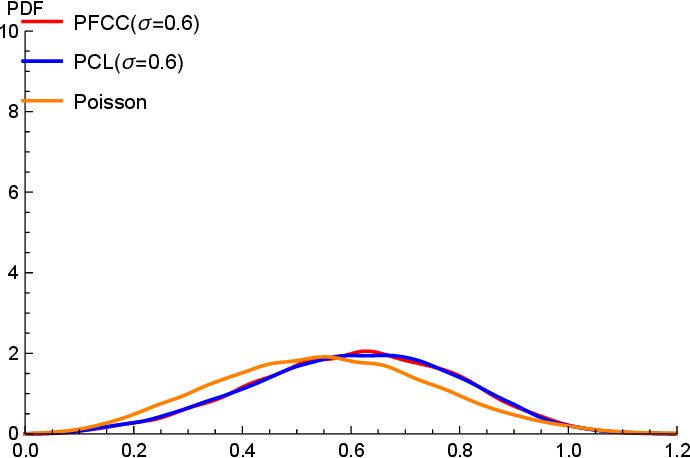}
\end{center}
\end{minipage}\\[4mm]
\begin{minipage}{0.52\hsize}
\begin{center}
\includegraphics[scale=\size4]{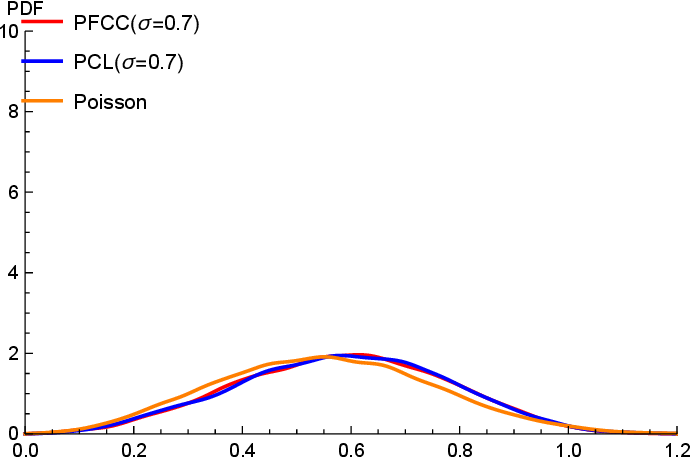}
\end{center}
\end{minipage}
\begin{minipage}{0.48\hsize}
\begin{center}
\includegraphics[scale=\size4]{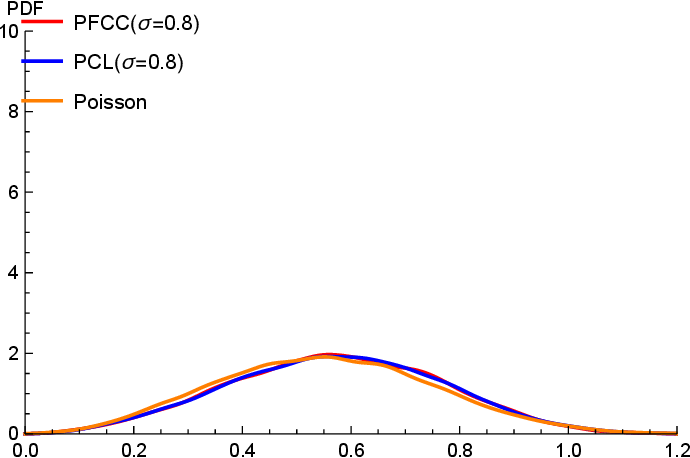}
\end{center}
\end{minipage}
\caption{Nearest neighbour distribution for perturbed Face-centered cubic lattice (PFCC), 
Perturbed Cubic Lattice(PCL) and Poisson ($\sigma=0.1, 0.2, \dots, 0.8$) for $\beta=2$. }
\label{fig:3dim-nearest-neighbour}
\end{figure}

\newpage

\bibliographystyle{plain}
\bibliography{PLN-3.bib}

\end{document}